%% file: report-en.tex
\documentclass[%
reprint,
superscriptaddress,
amsmath,amssymb,
aps,
amsthm,
pra,
]{revtex4-1}

\usepackage{graphicx}
\usepackage{dcolumn}
\usepackage{bm}
\usepackage{braket}%
\usepackage{theorem}
\usepackage{multirow}
\usepackage{enumerate}
\usepackage{framed}
\usepackage{bbm}

\usepackage[dvipdfmx,colorlinks,citecolor=blue,linkcolor=blue,urlcolor=blue,bookmarksopenlevel=4]{hyperref}

\usepackage{color}

\definecolor{memo}{RGB}{128,0,255}
\definecolor{gray}{RGB}{128,128,128}

\theorembodyfont{\upshape}

\input{settings_quant.tex}

\renewcommand{\c}{\circ}
\renewcommand{\ol}{\overline}
\newcommand{\opt}{\star}
\newcommand{\Realp}{\Real_+}
\newcommand{\g}{{(g)}}
\newcommand{\A}{{\rm A}}
\newcommand{\B}{{\rm B}}
\newcommand{\Pm}{{\rm P_m}}

\newcommand{\w}{{(\omega)}}
\newcommand{\tA}{\tilde{A}}
\newcommand{\PiA}{\Pi^{(\hA)}}
\newcommand{\hPiA}{\hPi^{(\hA)}}
\newcommand{\tAw}{\tA_\omega}
\newcommand{\hBw}{\hB^\w}
\newcommand{\hsw}{\hsigma_\omega}
\newcommand{\hxw}{\hx_\omega}
\newcommand{\dA}{d_\A}
\newcommand{\summ}{\sum_{m=0}^{M-1}}
\newcommand{\sumj}{\sum_{j=0}^{J-1}}
\newcommand{\sumg}{\sum_{g \in \mG}}
\newcommand{\sumr}{\sum_{r=0}^{R-1}}
\newcommand{\ts}{\tilde{s}}
\newcommand{\POVMA}{\POVM_\A}
\newcommand{\POVMB}{\POVM_\B}
\newcommand{\oPOVMA}{\ol{\POVMA}}
\newcommand{\dw}{d\omega}
\newcommand{\hAw}{\hA(\omega)}
\newcommand{\hAdw}{\hA(\dw)}
\newcommand{\pinc}{p_{\rm I}}
\newcommand{\TrB}{\Tr_\B}
\newcommand{\mHA}{\mH_\A}
\newcommand{\mHB}{\mH_\B}
\newcommand{\mSA}{\mS_\A}
\newcommand{\mSB}{\mS_\B}
\newcommand{\identA}{\ident_\A}
\newcommand{\identB}{\ident_\B}
\newcommand{\Prob}{\mP}
\newcommand{\inv}[1]{\overline{#1}}

\begin{document}

\preprint{APS/123-QED}

\title{Generalized bipartite quantum state discrimination problems with sequential measurements}

\affiliation{%
 Quantum Information Science Research Center, Quantum ICT Research Institute, Tamagawa University,
 Machida, Tokyo 194-8610, Japan
}%
\affiliation{%
 School of Information Science and Technology,
 Aichi Prefectural University,
 Nagakute, Aichi 480-1198, Japan
}%

\author{Kenji Nakahira}
\affiliation{%
 Quantum Information Science Research Center, Quantum ICT Research Institute, Tamagawa University,
 Machida, Tokyo 194-8610, Japan
}%

\author{Kentaro Kato}
\affiliation{%
 Quantum Information Science Research Center, Quantum ICT Research Institute, Tamagawa University,
 Machida, Tokyo 194-8610, Japan
}%

\author{Tsuyoshi \surname{Sasaki Usuda}}
\affiliation{%
 School of Information Science and Technology,
 Aichi Prefectural University,
 Nagakute, Aichi 480-1198, Japan
}%
\affiliation{%
 Quantum Information Science Research Center, Quantum ICT Research Institute, Tamagawa University,
 Machida, Tokyo 194-8610, Japan
}%

\date{\today}

\begin{abstract}
 We investigate an optimization problem of finding quantum sequential measurements,
 which forms a wide class of state discrimination problems
 with the restriction that only sequential measurements are allowed.
 Sequential measurements from Alice to Bob on a bipartite system are considered.
 Using the fact that the optimization problem can be formulated as
 a problem with only Alice's measurement and is convex programming,
 we derive its dual problem and necessary and sufficient conditions for an optimal solution.
 In the problem we address, the output of Alice's measurement can be infinite or continuous,
 while sequential measurements with a finite number of outcomes are considered.
 It is shown that there exists an optimal sequential measurement in which Alice's measurement
 with a finite number of outcomes as long as a solution exists.
 We also show that if the problem has a certain symmetry,
 then there exists an optimal solution with the same type of symmetry.
 A minimax version of the problem is considered,
 and necessary and sufficient conditions for a minimax solution are derived.
 An example in which our results can be used to obtain an analytical expression
 for an optimal sequential measurement is finally provided.
\end{abstract}

\maketitle

\section{Introduction} \label{sec:intro}

The study of the power and limitations of local discrimination of quantum states
has attracted considerable interest in quantum information theory in recent years.
In particular, sequential measurements,
which can be implemented using local measurements and one-way classical communication
(one-way LOCC), have been widely investigated.
Sequential measurements are relatively easy to implement with current technology;
for example, when two or more parties receive quantum states at different times,
measurements in which individual measurements are performed sequentially
would be desirable in practical implementations of quantum measurements.
However, it is well known that orthogonal quantum states shared by separated parties
may not be perfectly distinguished when only sequential measurements are allowed,
while they can be perfectly distinguished by a global measurement.
This implies that sequential measurements are less powerful than global measurements
for quantum state discrimination.
An important question that arises in studies of this kind is
how well one can distinguish between given quantum states by a sequential measurement.

Many studies have been developed to tackle the problem of which sets of orthogonal states
are distinguishable when only sequential measurements are allowed
(e.g., \cite{Wal-Sho-Har-Ved-2000,Gro-Vai-2001,Wal-Har-2002,Fan-2004,Nat-2005,Ban-Gho-Kar-2011,Zha-Wen-Gao-Tia-2014}).
There have also been several investigations of a sequential measurement
realizing a measurement that maximizes the average success probability
(called a minimum-error measurement)
\cite{Vir-Sac-Ple-Mar-2001,Bro-Mei-1996,Aci-Bag-Bai-Mas-Mun-2004,Ass-Poz-Pie-2011}.
It has also been reported that a measurement that maximizes the average success probability
with no error at the expense of allowing for a certain fraction of inconclusive (failure) results
(called an optimal unambiguous measurement)
can be realized by a sequential measurement for binary pure states
\cite{Che-Yan-2001,Che-Yan-2002,Ji-Cao-Yin-2005}.
However, these results are only applicable to a special class of quantum states.
Investigations applicable to a broad class of quantum states would be required.

In the scenario in which all quantum measurements are allowed,
optimal measurement strategies have been investigated under various criteria,
such as the Bayes criterion \cite{Hol-1973,Yue-Ken-Lax-1975,Hel-1976}
and the minimax criterion \cite{Hir-Ike-1982,Dar-Sac-Kah-2005,Kat-2012}.
A measurement strategy that allows for inconclusive results
has also been well studied.
The most well-known example along this line is an optimal unambiguous measurement
\cite{Iva-1987,Die-1988,Per-1988}.
Other examples are a measurement that maximizes the average success probability
with a fixed average inconclusive probability,
denoted as an optimal inconclusive measurement
\cite{Che-Bar-1998-inc,Eld-2003-inc,Fiu-Jez-2003},
and a measurement that maximizes the average success probability
under the condition that the average error probability should not exceed a certain error,
denoted as an optimal error margin measurement
\cite{Tou-Ada-Ste-2007,Hay-Has-Hor-2008,Sug-Has-Hor-Hay-2009}.
Recently, a generalized state discrimination problem,
which is applicable to the above mentioned criteria, has also been presented \cite{Nak-Kat-Usu-2015-general}.
From these studies, some properties of optimal measurements in the above criteria,
such as necessary and sufficient conditions for optimality,
have been derived.
By contrast, in the case of a sequential measurement,
very few studies of an optimal sequential measurement
for a strategy other than
the minimum error strategy and the unambiguous strategy have been reported
(e.g, \cite{Ban-Yam-Hir-1997,Owa-Hay-2008,Nak-Usu-2012-receiver,Nak-Usu-2016-LOCC,Ros-Mar-Gio-2017-capacity}).

More recently, Croke {\it et al.} have derived a necessary and sufficient condition
for a sequential measurement to maximize the average success probability
(we call such a measurement a minimum-error sequential measurement)
and used it to prove optimality of a candidate solution \cite{Cro-Bar-Wei-2017}.
Also, the authors have derived
the dual problem of the problem of finding a minimum-error sequential measurement
and utilized it to compute numerical solutions \cite{Nak-Kat-Usu-2017-Dolinar}.
These results are applicable to arbitrary bipartite quantum states;
however, only a few properties of a minimum-error sequential measurement have ever been reported.
In addition, these methods cannot directly be applied to other criteria.

In this paper, we address a sequential-measurement version of the generalized
state discrimination problem described in Ref.~\cite{Nak-Kat-Usu-2015-general}.
Similarly as in this reference, this problem includes problems with various criteria.
We consider sequential measurements from Alice to Bob on a bipartite system.
Since the problem of finding an optimal sequential measurement is
much more complex than that of finding an optimal global measurement,
the results proposed in Ref.~\cite{Nak-Kat-Usu-2015-general}
cannot readily be applied to this problem.
However, we can see that the entire set of sequential measurements is convex;
thus, the generalized state discrimination problem with sequential measurements
can be formulated as a convex programming problem.
Useful results available in convex programming help us to further understand
an optimal sequential measurement.
In the problem we address, sequential measurements with a finite number of outcomes are considered,
whereas the output of Alice's measurement can be infinite or continuous.
We show that there always exists an optimal sequential measurement in which Alice's measurement
with a finite number of outcomes as long as a solution exists.
We also derive the dual problem of the original problem
and necessary and sufficient conditions for an optimal solution.
These properties would be useful to obtain analytical and numerical expressions for
an optimal sequential measurement.

In Sec.~\ref{sec:gen}, we discuss the formulation of sequential measurements
and provide a sequential-measurement version of the generalized state discrimination problem.
In Sec.~\ref{sec:gen_opt}, its dual problem is derived.
Then, we show that the optimal values of the primal and dual problems are the same.
Necessary and sufficient conditions for an optimal solution is also obtained.
In Sec.~\ref{sec:sym}, we show that if a problem has a certain symmetry,
then there exists an optimal solution with the same type of symmetry.
In Sec.~\ref{sec:minimax}, we discuss a sequential-measurement version of
the generalized minimax problem described in Ref.~\cite{Nak-Kat-Usu-2015-general}.
We also derive necessary and sufficient conditions for a minimax solution.
In Sec.~\ref{sec:example}, as an example, our results are applied to the problem of
finding an optimal inconclusive sequential measurement.
An analytical expression of an optimal inconclusive sequential measurement for double trine states
is also derived.
This example illustrates that our results can be used to obtain an analytical solution
to at least an easy problem.

\section{Generalized optimal sequential measurement} \label{sec:gen}

\subsection{Sequential measurement}

We consider a composite system, $\mH = \mHA \otimes \mHB$, of two subsystems, Alice and Bob.
Let $\mS$ and $\mS^+$ be, respectively, the entire sets of
Hermitian operators and positive semidefinite operators on $\mH$.
$\mS_k$ and $\mS_k^+$ $(k \in \{ \A, B \})$ are defined in the same way
with $\mH$ replaced by $\mH_k$.
Also, let $\Real$ and $\Real_+$ be, respectively, the entire sets of real numbers
and nonnegative real numbers,
and $\mI_N \equiv \{ 0, 1, \cdots, N-1 \}$.
Let $\ident$, $\identA$, and $\identB$ be, respectively,
the identity operators on $\mH$, $\mHA$, and $\mHB$.
We denote $\{ tb_n \}$ and $\{ b_n + b'_n \}$ with $t \in \Real$ and
$b, b' \in \Real^N$ (or $b, b' \in \Realp^N$) as $tb$ and $b + b'$, respectively.
$\hx \ge \hy$ with Hermitian operators $\hx$ and $\hy$ denotes that $\hx - \hy$ is positive semidefinite.

Let us consider a sequential measurement on $\mH$.
Alice first performs a measurement, which is represented by
a positive operator valued measure (POVM) $\{ \hA_j \in \mSA^+ \}_j$,
the output of which can be infinite (or continuous).
The measurement result $j$ is sent to Bob.
Then, Bob chooses a measurement $\{ \hB^{(j)}_m \in \mSB^+ \}_{m=0}^{M-1}$ depending on $j$,
and obtains the outcome $m \in \mI_M$,
which represents the final measurement result.
The measurement on the joint system is given by the POVM
$\{ \hPi_m = \sum_j \hA_j \otimes \hB^{(j)}_m \}_{m=0}^{M-1}$.

We can consider this sequential measurement from a different viewpoint
\cite{Nak-Kat-Usu-2017-Dolinar}.
Let $\POVMB$ be the entire set of allowed Bob's measurements
and $\Omega$ be an isomorphic set of $\POVMB$.
Each element of $\POVMB$ is uniquely labeled by an index $\omega \in \Omega$;
we denote Bob's measurement corresponding to $\omega \in \Omega$ as $\hBw \equiv \{ \hBw_m \}_{m=0}^{M-1}$.
Alice first performs a measurement, $\hA$, with continuous outcomes in $\Omega$,
She sends the result $\omega \in \Omega$ to Bob.
He performs the corresponding measurement $\hBw$.
Alice's POVM $\hA$ uniquely determines this sequential measurement,
which is denoted as $\PiA \equiv \{ \hPiA_m \}_{m=0}^{M-1}$ with
\begin{eqnarray}
 \hPiA_m &\equiv& \int_{\Omega} \hAdw \otimes \hBw_m. \label{eq:PiA}
\end{eqnarray}
We can interpret that Alice's POVM, $\hA$, includes all the information regarding
the measurements Bob should perform.
Let $\POVMA$ be the entire set of Alice's POVMs.
Any sequential measurement can be denoted as $\PiA$ with $\hA \in \POVMA$.
In this formulation, the problem of finding an optimal sequential measurement
can be formulated as an optimization problem with only $\hA$.

Let $\sigma(\Omega)$ be the sigma algebra of all measurable subsets of $\Omega$.
$\hA \in \POVMA$ is a mapping of $\sigma(\Omega)$ into $\mSA^+$,
which satisfies
\begin{enumerate}[(1)]
 \setlength{\itemsep}{0pt}
 \setlength{\parskip}{0pt}
 \item positivity: $\hPhi(E) \ge 0, ~ \forall E \in \sigma(\Omega)$,
 \item countable additivity: $\hPhi(\cup_k E_k) = \sum_k \hPhi(E_k)$ with mutually disjoint
       $\{ E_k \} \subset \sigma(\Omega)$,
 \item normalization: $\hPhi(\Omega) = \identA$.
\end{enumerate}
Let $\oPOVMA$ be the entire set of (not necessarily normalized) mappings
$\hA: \sigma(\Omega) \to \mSA^+$ satisfying
the conditions (1) and (2).
Obviously, $\oPOVMA \supset \POVMA$ holds.

It should be noted that $\POVMB$ is not necessarily the entire set
of POVMs on $\mHB$; for example, $\mHB$ can be a composite system of $n$ subsystems,
and $\POVMB$ can be the entire set of sequential measurements (or two-way LOCC measurements) on $\mHB$.

\subsection{State discrimination problem}

Here, we consider a sequential-measurement version of the optimization problem
described in Ref.~\cite{Nak-Kat-Usu-2015-general},
which is expressed as
\begin{eqnarray}
 \begin{array}{lll}
  {\rm P:} & {\rm maximize} & \displaystyle f(\hA) \equiv \summ \Tr \left[ \hc_m\hPiA_m \right] \\
  & {\rm subject~to} & \hA \in \POVMA^\c, \\
 \end{array} \label{eq:primal}
\end{eqnarray}
where
\begin{eqnarray}
 \POVMA^\c &\equiv& \left\{ \hA \in \POVMA : \eta_j(\hA) \le 0, ~ \forall j \in \mI_J \right\}, \nonumber \\
 \eta_j(\hA) &\equiv& \summ \Tr \left[ \ha_{j,m} \hPiA_m \right] - b_j,
  \label{eq:POVMc}
\end{eqnarray}
$\hc_m \in \mS$, $\ha_{j,m} \in \mS$, and $b_j \in \Real$.
$J$ is a nonnegative integer that represents the number of constraints.

As an example, let us consider the problem of obtaining a minimum-error sequential measurement
for the states $\{ \trho_m \}_{m=0}^{M-1}$ with equal prior probabilities $\{ \xi_m \}_{m=0}^{M-1}$,
which is expressed as
\begin{eqnarray}
 \begin{array}{ll}
  {\rm maximize} & \displaystyle \summ \Tr \left[ \hrho_m \hPiA_m \right] \\
  {\rm subject~to} & \hA \in \POVMA, \\
 \end{array} \label{eq:primal_MEM}
\end{eqnarray}
where $\hrho_m = \xi_m \trho_m$.
This problem is obtained by substituting $\hc_m = \hrho_m$ and $J = 0$ into Problem~P.
Problem~P can express a large class of problems;
one can find some examples in Subsec. II.B of Ref.~\cite{Nak-Kat-Usu-2015-general}
(also, see Sec.~\ref{sec:example} of this paper).

We can easily verify that $\POVMA^\c$ is convex,
and thus Problem~P is a convex programming.
Let $f^\opt$ be the optimal value of Problem~P.
$f^\opt$ is regarded as $-\infty$ if the feasible set, $\POVMA^\c$, is empty.
Note that an equality constraint, $\eta_j(\hA) = 0$, can be replaced
by two inequality constraints, $\eta_j(\hA) \le 0$ and $-\eta_j(\hA) \le 0$.

\section{Optimal solution to generalized problem} \label{sec:gen_opt}

\subsection{Dual problem}

We will derive the dual problem of Problem~P, which is formulated as follows:
\begin{eqnarray}
 \begin{array}{lll}
  {\rm DP:} & {\rm minimize} & \displaystyle s(\hX, \lambda) \equiv \Tr~\hX + \sumj \lambda_j b_j \\
  & {\rm subject~to} & (\hX, \lambda) \in \mX^\c \\
 \end{array}
 \label{eq:dual}
\end{eqnarray}
with variables $\hX$ and $\lambda$, where
\begin{eqnarray}
 \mX^\c &\equiv& \left\{ (\hX, \lambda) \in \mX : \hX \ge \hsw(\lambda),
               ~ \forall \omega \in \Omega \right\}, \nonumber \\
 \mX &\equiv& \mSA \otimes \Realp^J, \nonumber \\
 \hsw(\lambda) &\equiv& \TrB \summ \hz_m(\lambda) \hBw_m, \nonumber \\
 \hz_m(\lambda) &\equiv& \hc_m - \sumj \lambda_j \ha_{j,m}. \label{eq:z}
\end{eqnarray}
$\TrB$ is the partial trace with respect to the system $\mHB$.
Let $s^\opt$ be the optimal value of Problem~DP.

We define the following Lagrangian for Problem~P as:
\begin{eqnarray}
 L(\hA, \hX, \lambda) &\equiv& f(\hA) + \Tr[ \hX [\identA - \hA(\Omega)]] \nonumber \\
 & & \mbox{} - \sumj \lambda_j \eta_j(\hA), \label{eq:L}
\end{eqnarray}
where $L(\hA, \hX, \lambda)$ is a function of $\hA \in \oPOVMA$ and $(\hX, \lambda) \in \mX$.
If $\hA(\Omega) \neq \identA$ holds, then
there exists a vector $\ket{x}$ satisfying $\braket{x | [\identA - \hA(\Omega)] | x} \neq 0$;
taking the limit $t \to \infty$ or $t \to -\infty$
yields $L(\hA, t \ket{x}\bra{x}, \lambda) \to - \infty$.
Similarly, if there exists $j \in \mI_J$ such that $\eta_j(\hA) > 0$,
then $L(\hA, \hX, \lambda) \to -\infty$ when $\lambda_j \to \infty$.
Thus, if $\hA \not\in \POVMA^\c$ holds, then
there exists $(\hX, \lambda) \in \mX$ such that $L(\hA, \hX, \lambda) \to - \infty$.
On the other hand, if $\hA \in \POVMA^\c$ holds, then
$L(\hA, \hX, \lambda) \ge f(\hA)$ holds and the equality holds if $\lambda = 0$ holds.
Therefore, we obtain
\begin{eqnarray}
 \max_{\hA \in \oPOVMA} \min_{(\hX,\lambda) \in \mX} L(\hA, \hX, \lambda)
  &=& \max_{\hA \in \POVMA^\c} \min_{(\hX,\lambda) \in \mX} L(\hA, \hX, \lambda)
  \nonumber \\
 &=& \max_{\hA \in \POVMA^\c} f(\hA) = f^\opt.
  \label{eq:dual_max_min}
\end{eqnarray}

Let
\begin{eqnarray}
 \ts(\hX, \lambda) &\equiv& \max_{\hA \in \oPOVMA} L(\hA, \hX, \lambda). \label{eq:ts}
\end{eqnarray}
Substituting $F = f$, $x = \hA$, and $y = (\hX, \lambda)$ into the following formula:
\begin{eqnarray}
 \min_y \max_x F(x,y) \ge \max_x \min_y F(x,y)
\end{eqnarray}
and using Eqs.~\eqref{eq:dual_max_min} and \eqref{eq:ts} yields
\begin{eqnarray}
 \min_{(\hX,\lambda) \in \mX} \ts(\hX, \lambda) &\ge& f^\opt. \label{eq:s_P}
\end{eqnarray}
Let us consider the problem of finding $(\hX, \lambda) \in \mX$ that
minimizes $\ts(\hX, \lambda)$, which can be regarded as a dual problem of Problem~P.
From Eqs.~\eqref{eq:dual}--\eqref{eq:L}, $L(\hA, \hX, \lambda)$ is rewritten as
\begin{eqnarray}
 L(\hA, \hX, \lambda) &=& s(\hX, \lambda) + \int_\Omega \Tr [\hsw(\lambda) - \hX] \hAdw.
  \label{eq:L2}
\end{eqnarray}
If $(\hX, \lambda) \not\in \mX^\c$ holds (i.e., there exists $\omega$
such that $\hX \not\ge \hsw(\lambda)$), then
there exists a vector $\ket{x} \in \mHA$ such that $\braket{x | [\hX - \hsw(\lambda)] | x} < 0$;
substituting $\hAw = t \ket{x} \bra{x}$ into Eq.~\eqref{eq:L2}
and taking the limit $t \to \infty$ gives $L(\hA, \hX, \lambda) = \infty$.
Thus, from Eq.~\eqref{eq:ts}, $\ts(\hX, \lambda) = \infty$ holds.
On the other hand, if $(\hX, \lambda) \in \mX^\c$ holds, then
$L(\hA, \hX, \lambda)$ reaches its maximum value of $s(\hX, \lambda)$
when $\hA(E) = 0$ for any $E \subseteq \Omega$,
and thus $\ts(\hX, \lambda) = s(\hX, \lambda)$ holds.
Therefore, we obtain
\begin{eqnarray}
 \min_{(\hX,\lambda) \in \mX} \ts(\hX, \lambda)
  &=& \min_{(\hX, \lambda) \in \mX^\c} s(\hX, \lambda),
\end{eqnarray}
which indicates that the dual problem can be rewritten as Problem~DP.
From Eq.~\eqref{eq:s_P}, $s^\opt \ge f^\opt$ holds.

In a convex optimization problem,
the optimal values of the primal and dual problems are generally not the same.
However, as stated in the following theorem,
the optimal values of Problems~P and DP are always the same
(proof in Appendix~\ref{append:zero_gap}).
\begin{thm} \label{thm:zero_gap}
 $s^\opt = f^\opt$ always holds.
\end{thm}

\subsection{Conditions for an optimal solution}

In generalized state discrimination problems with no restriction on measurements,
necessary and sufficient conditions for an optimal solution have been derived
\cite{Nak-Kat-Usu-2015-general}.
In a similar manner, we can derive necessary and sufficient conditions
for an optimal solution to Problem~P using its dual problem.
\begin{thm} \label{thm:condition}
 Let $\hA$ be a POVM satisfying $\hA \in \POVMA^\c$.
 The following statements are all equivalent.
 \begin{enumerate}[(1)]
  \setlength{\parskip}{0cm}
  \setlength{\itemsep}{0cm}
  \item $\hA$ is an optimal solution to Problem~P.
  \item There exists $(\hX, \lambda) \in \mX^\c$ such that
		\begin{eqnarray}
		 [\hX - \hsw(\lambda)] \hAw &=& 0, ~ \forall \omega \in \Omega, \label{eq:cond_XCA} \\
		 \lambda_j \eta_j(\hA) &=& 0, ~ \forall j \in \mI_J. \label{eq:cond_b}
		\end{eqnarray}
  \item There exists $\lambda \in \Realp^J$ such that
		\begin{eqnarray}
		 \int_\Omega \hsigma_{\omega'}(\lambda) \hA(\dw') &\ge& \hsw(\lambda),
          ~ \forall \omega \in \Omega, \label{eq:cond_Xlambda} \\
		 \lambda_j \eta_j(\hA) &=& 0, ~ \forall j \in \mI_J. \label{eq:cond_b2}
		\end{eqnarray}
 \end{enumerate}
 Moreover, if Condition~(2) holds, then $(\hX, \lambda)$ is an optimal solution to Problem~DP.
\end{thm}

From Eq.~\eqref{eq:cond_XCA}, for any $\omega \in \Omega$,
the kernel of $\hX - \hsw(\lambda)$ includes the support of $\hAw$.
Note that Condition~(3) in the case of the problem of obtaining a minimum-error sequential measurement
is given in Ref.~\cite{Cro-Bar-Wei-2017}.

\begin{proof}
 We will show (1) $\Rightarrow$ (2), (2) $\Rightarrow$ (3), and (3) $\Rightarrow$ (1) in this order.
 After that, we will show that $(\hX, \lambda)$ is an optimal solution to Problem~DP if Condition~(2) holds.

 First, we show (1) $\Rightarrow$ (2).
 Let $(\hX, \lambda)$ be an optimal solution to Problem~DP.
 Since $\hA(\Omega) = \identA$ and $\eta_j(\hA) \le 0$ hold,
 the second and third terms of the right-hand side of Eq.~\eqref{eq:L} are
 zero and nonnegative, respectively,
 which gives $L(\hA, \hX, \lambda) \ge f(\hA) = f^\opt$.
 Also, since $\hX \ge \hsw(\lambda)$ and $\hA(\omega) \ge 0$ hold,
 the second term of the right-hand side of Eq.~\eqref{eq:L2} is
 nonpositive, which gives $L(\hA, \hX, \lambda) \le s(\hX, \lambda) = s^\opt$ holds.
 Since $f^\opt = s^\opt$ holds from Theorem~\ref{thm:zero_gap}, we obtain
 \begin{eqnarray}
  f^\opt &=& L(\hA, \hX, \lambda) = s^\opt,
 \end{eqnarray}
 i.e., the third term of the right-hand side of Eq.~\eqref{eq:L}
 and the second term of the right-hand side of Eq.~\eqref{eq:L2}
 must be zero.
 Thus, Eqs.~\eqref{eq:cond_XCA} and \eqref{eq:cond_b} hold.
 Note that Eq.~\eqref{eq:cond_XCA} follows from the fact that $\hx\hy = 0$ holds for
 any $\hx, \hy \in \mSA^+$ satisfying $\Tr(\hx\hy) = 0$.

 Next, we show (2) $\Rightarrow$ (3).
 Integrating both sides of Eq.~\eqref{eq:cond_XCA} and using $\hA(\Omega) = \identA$ gives
 \begin{eqnarray}
  \hX &=& \int_\Omega \hsw(\lambda) \hAdw. \label{eq:cond_X}
 \end{eqnarray}
 $\hX \ge \hsw(\lambda)$ gives Eq.~\eqref{eq:cond_Xlambda}.
 Equation~\eqref{eq:cond_b2} is equivalent to Eq.~\eqref{eq:cond_b}.

 We show (3) $\Rightarrow$ (1).
 We define $\hX$ as in Eq.~\eqref{eq:cond_X}.
 We have that for any POVM $\hA' \in \POVMA^\c$,
 \begin{eqnarray}
  \lefteqn{ f(\hA) - f(\hA') } \nonumber \\
  &\ge& f(\hA) - \sumj \lambda_j \eta_j(\hA) - f(\hA') + \sumj \lambda_j \eta_j(\hA') \nonumber \\
  &=& \summ \Tr \left[ \hz_m(\lambda) \hPiA_m - \hz_m(\lambda) \hPi^{(\hA')}_m \right] \nonumber \\
  &=& \Tr~\hX - \Tr~\int_\Omega \hsw(\lambda) \hA'(\dw) \nonumber \\
  &=& \Tr \int_\Omega [\hX - \hsw(\lambda)] \hA'(\dw) \ge 0. \label{eq:cond_fA_fA_}
 \end{eqnarray}
 The second line follows from Eq.~\eqref{eq:cond_b2} and $\eta_j(\hA') \le 0$.
 The third line follows from Eqs.~\eqref{eq:primal}, \eqref{eq:POVMc}, and \eqref{eq:z}.
 The fourth line follows from the fact that, from Eqs.~\eqref{eq:PiA} and \eqref{eq:z},
 we have that for any $\hPhi$,
 \begin{eqnarray}
  \summ \Tr \left[ \hz_m(\lambda) \hPi^{(\hPhi)}_m \right] &=& \Tr~\int_\Omega \hsw(\lambda) \hPhi(\dw).
 \end{eqnarray}
 The last inequality follows from Eq.~\eqref{eq:cond_Xlambda} (i.e., $\hX \ge \hsw(\lambda)$).
 From Eq.~\eqref{eq:cond_fA_fA_}, $\hA$ is an optimal solution to Problem~P.

 Finally, we will show that $(\hX, \lambda)$ is an optimal solution to Problem~DP if Condition~(2) holds.
 From Eqs.~\eqref{eq:L2} and \eqref{eq:cond_XCA}, $L(\hA, \hX, \lambda) = s(\hX, \lambda)$ holds.
 Also, from Eqs.~~\eqref{eq:L} and \eqref{eq:cond_b}, $L(\hA, \hX, \lambda) = f(\hA) = f^\opt$ holds.
 Thus, $s(\hX, \lambda) = f^\opt$ holds, which means that $(\hX, \lambda)$
 is an optimal solution to Problem~DP.
 \QED
\end{proof}

We should mention that obtaining an optimal solution to Problem~P
is much more difficult than obtaining an optimal solution to
the problem described in Ref.~\cite{Nak-Kat-Usu-2015-general},
i.e., the state discrimination problem with no restriction on measurements.
The reason is that, in the former case, we have to optimize over all of Alice's measurements,
which include all the information regarding the measurements Bob should perform.
Problem~DP is generally difficult to solve as well as Problem~P.
However, we can obtain an analytical solution by solving Problem~DP in some cases
(see Subsec.~\ref{subsec:dtrine}).

\subsection{Number of outcomes of Alice's POVM}

So far in this paper, we have considered Alice's POVM $\hA$ to be continuous.
We find that an optimal solution to Problem~P with finite outcomes
always exists as long as a feasible solution exists,
as shown in the following theorem
(proof in Appendix~\ref{append:finite}):

\begin{thm} \label{thm:finite}
 Let $\dA = \dim~\mHA$.
 If $\POVMA^\c$ is not empty, then
 an optimal solution to Problem~P with at most $(J+1) \dA^2$ outcomes exists.
\end{thm}

\subsection{Comparison with the problem with no restriction on measurements}

Table~\ref{tab:summary} summarizes the formulation
of the state discrimination problems (a) when arbitrary measurements are allowed
and (b) when only sequential measurements are allowed.
The dual problem in the case (b) (i.e., Problem~DP) has
an infinite (continuous) number of constraints,
while that in the case (a) has a finite number $M$ of constraints.
This makes it difficult to obtain an optimal sequential measurement.

\begin{table*}[tbp]
 \caption{Formulation of the generalized state discrimination problems.}
 \label{tab:summary}
 \begin{center}
  \tabcolsep = 0.4em
  \begin{tabular}{p{0.5em}ll}
   \hline \hline
   & \multicolumn{1}{c}{(a) Arbitrary measurements \cite{Nak-Kat-Usu-2015-general}}
       & \multicolumn{1}{c}{(b) Sequential measurements} \\ \hline
   \multicolumn{3}{l}{Primal problems} \\
   & \raisebox{-0.0em}{\shortstack[l]{
     ${\rm maximize} ~ \displaystyle \summ \Tr(\hc_m \hPi_m)$ \\
     ${\rm subject~to} ~ \Pi {\rm ~ : POVM}$, \\
     \hspace{1em} $\displaystyle \summ \Tr(\ha_{j,m} \hPi_m) \le b_j ~ (\forall j \in \mI_J)$ }}
   & \raisebox{-2.5em}{\shortstack[l]{
     ${\rm maximize} ~ \displaystyle \summ \Tr \left[ \hc_m\hPiA_m \right]$ \\
     ${\rm subject~to} ~ \hA \in \POVMA$, \\
     \hspace{1em} $\displaystyle \summ \Tr \left[ \ha_{j,m} \hPiA_m \right] \le b_j ~ (\forall j \in \mI_J)$ \\
     \hfill \eqref{eq:primal},\eqref{eq:POVMc}}} \\
   \multicolumn{3}{l}{Dual problems} \\
   & \raisebox{-0.0em}{\shortstack[l]{
     ${\rm minimize} ~~ \displaystyle \Tr~\hX + \sumj \lambda_j b_j$ \\
     ${\rm subject~to} ~ \hX \ge \hz_m(\lambda) ~ (\forall m \in \mI_M), \lambda \in \Realp^J$ \\
     where ~ $\displaystyle \hz_m(\lambda) = \hc_m - \sumj \lambda_j \ha_{j,m}$ }}
   & \raisebox{-5.6em}{\shortstack[l]{
     ${\rm minimize} ~~ \displaystyle \Tr~\hX + \sumj \lambda_j b_j$ \\
     ${\rm subject~to} ~ \hX \ge \hsw(\lambda) ~ (\forall \omega \in \Omega), \lambda \in \Realp^J$ \\
     where ~ $\displaystyle \hsw(\lambda) = \TrB \summ \hz_m(\lambda) \hBw_m$, \\
     ~~~~~~  $\displaystyle \hz_m(\lambda) = \hc_m - \sumj \lambda_j \ha_{j,m}$ \\
     \hfill \eqref{eq:dual},\eqref{eq:z}}} \\
   \multicolumn{3}{l}{Necessary and sufficient conditions for optimality (Condition~(3))} \\
   & \raisebox{-0.0em}{\shortstack[l]{
     $\lambda \in \Realp^J$ exists such that \\
     \hspace{0.5em} $\displaystyle \summ \hz_m(\lambda) \hPi_m \ge \hz_m(\lambda), ~ \forall m \in \mI_M$, \\
     \hspace{0.5em} $\displaystyle \lambda_j \left[ b_j - \summ \Tr(\ha_{j,m} \hPi_m) \right] = 0,
                    ~ \forall j \in \mI_J$ }}
   & \raisebox{-1.5em}{\shortstack[l]{
     $\lambda \in \Realp^J$ exists such that \\
   \hspace{0.5em} $\displaystyle \int_\Omega \hsigma_{\omega'}(\lambda) \hA(\dw') \ge \hsw(\lambda),
                  ~ \forall \omega \in \Omega$, \\
     \hspace{0.5em} $\displaystyle \lambda_j \left[ b_j - \summ \Tr \left[ \ha_{j,m} \hPiA_m \right] \right] = 0,
                    ~ \forall j \in \mI_J$ \\
     \hfill \eqref{eq:cond_Xlambda},\eqref{eq:cond_b2}}} \\
   \hline \hline
  \end{tabular}
 \end{center}
\end{table*}

\section{Group covariant problem} \label{sec:sym}

In this section, we discuss the case in which Problem~P has a certain symmetry.
State discrimination problems with symmetries have been well studied,
and it is known that, in some cases, there exists an optimal solution with the same type of symmetry
\cite{Bel-1975,Ban-Kur-Mom-Hir-1997,Usu-Tak-Hat-Hir-1999,Eld-For-2001,Eld-2003-inc,Eld-Meg-Ver-2004,Eld-Sto-Has-2004,Nak-Usu-2013-group,Nak-Kat-Usu-2013-minimax}.
The existence of a symmetric solution helps us to obtain analytical or numerical optimal solutions
(e.g., \cite{And-Bar-Gil-Hun-2002,Kat-Hir-2003,Qiu-2008,Ass-Car-Pie-2010,Nak-Kat-Usu-2015-numerical}).

\subsection{Group action}

First, we briefly introduce a group action.
Let $\mG$ be a group and $e \in \mG$ be its identity element.
Also, let $\inv{g} \in \mG$ be the inverse element of $g \in \mG$.
We assume that $\mG$ has at least two elements.
Let $|\mG|$ be the number of elements in $\mG$.
A group action of $\mG$ on a set $T$ is a set of mappings on $T$,
$\{ \pi_g : T \to T \}_{g \in \mG}$, such that
\begin{eqnarray}
 \pi_{gh}(x) &=& \pi_g[\pi_h(x)], ~~ \forall g, h \in \mG, ~ x \in T, \nonumber \\
 \pi_e(x) &=& x, ~~ \forall x \in T. \label{eq:group_action0}
\end{eqnarray}
In what follows, we denote $\pi_g(x)$ as $g \c x$.
Equation~\eqref{eq:group_action0} can be rewritten by
\begin{eqnarray}
 (gh) \c x &=& g \c (h \c x), ~~ \forall g, h \in \mG, ~ x \in T, \nonumber \\
 e \c x &=& x, ~~ \forall x \in T. \label{eq:group_action}
\end{eqnarray}
The action is called faithful if, for any distinct $g,h \in \mG$,
there exists $x \in T$ such that $g \c x \neq h \c x$.

Next, we set actions of $\mG$ on the sets $\mI_N$, $\mS$, $\mSA$, $\mSB$, and $\Omega$
as follows.
An action of $\mG$ on $\mI_N$,
$\{ g \c n ~(n \in \mI_N) \}_{g \in \mG}$,
is given by a set of permutations of $\{ 0, \cdots, N-1 \}$,
which is not necessarily faithful.
We choose them such that they meet the conditions of Theorem~\ref{thm:sym} described below.

We also consider the action of $\mG$ on $\mS$ expressed by
\begin{eqnarray}
 g \c \hQ &\equiv& \hU_g \hQ \hU_g^\dagger, ~~ \forall g \in \mG, \hQ \in \mS, \label{eq:group_action_S}
\end{eqnarray}
where $\hU_g$ is a unitary or anti-unitary operator and $^\dagger$ is the conjugate transpose operator.
Note that if $\hU_g$ is an anti-unitary operator, then
$\hU_g^\dagger$ is also anti-unitary such that $\hU_g^\dagger \hU_g = \hU_g \hU_g^\dagger = \ident$.
From Eq.\eqref{eq:group_action}, $\hU_{gh}$ equals $\hU_g \hU_h$ up to a global phase for any $g, h \in \mG$,
and $\hU_e = \ident$ holds.
Assume that the action of $\mG$ on $\mS$ is faithful,
i.e., $\hU_g$ and $\hU_h$ are not equivalent up to a global phase for any distinct $g, h \in \mG$.
Also, assume that $\hU_g$ can be expressed by
\begin{eqnarray}
 \hU_g &=& \hV_g \otimes \hW_g, \label{eq:UVW}
\end{eqnarray}
where $\hV_g$ and $\hW_g$ are, respectively, unitary or anti-unitary operators on $\mHA$ and $\mHB$.
We can easily verify that $\hV_{gh}$ and $\hW_{gh}$, respectively, equal
$\hV_g \hV_h$ and $\hW_g \hW_h$ up to global phases for any $g, h \in \mG$,
and $\hV_e = \identA$ and $\hW_e = \identB$ hold.

We set actions of $\mG$ on $\mSA$ and $\mSB$ as follows:
\begin{eqnarray}
 g \c \hQ^{(\A)} &\equiv& \hV_g \hQ^{(\A)} \hV_g^\dagger,
  ~~ \forall g \in \mG, \hQ^{(\A)} \in \mSA, \nonumber \\
 g \c \hQ^{(\B)} &\equiv& \hW_g \hQ^{(\B)} \hW_g^\dagger,
  ~~ \forall g \in \mG, \hQ^{(\B)} \in \mSB.
  \label{eq:group_action_SA}
\end{eqnarray}
These actions are not necessarily faithful.

We stress that actions of $\mG$ are different among different sets.
For example, $g \c \hQ$ with $\hQ \in \mS$ and $g \c \hQ^{(\A)}$ with $\hQ^{(\A)} \in \mSA$
mean $\hU_g \hQ \hU_g^\dagger$ and $\hV_g \hQ^{(\A)} \hV_g^\dagger$, respectively.

Assume that $\{ g \c \hBw_m \}_m$ is in $\POVMB$ for any $g \in \mG$ and $\omega \in \Omega$
\footnote{
This assumption always holds if $\POVMB$ is the entire set of POVMs on $\mHB$;
otherwise, it does not hold in general.
For example, if $\mHB$ is a composite system and $\POVMB$ is the entire set of
sequential measurements on $\mHB$, then $\{ g \c \hBw_m \}_m$ might not be in $\POVMB$
in spite of $\{ \hBw_m \}_m \in \POVMB$.
In such cases, we need to appropriately set the action of $\mG$ on $\mSB$.}.
We set an action of $\mG$ on $\Omega$, $\{ g \c \omega ~(\omega \in \Omega) \}_{g \in \mG}$, such that
\begin{eqnarray}
 g \c \hBw_m &=& \hB^{(g \c \omega)}_{g \c m}, ~~ \forall g \in \mG, m \in \mI_M, \omega \in \Omega.
  \label{eq:sym_B}
\end{eqnarray}

The above model can handle various symmetries.
For example, in the case in which only Bob's system has a certain symmetry,
we can consider a group $\mG$ with $\hV_g = \identA$ for any $g \in \mG$.
As another example,
if Alice's and Bob's systems independently have different symmetries,
represented by groups $\mG_\A$ and $\mG_\B$ respectively,
then we can consider the direct product of the groups, $\mG = \mG_\A \times \mG_\B$;
we can define the actions of $\mG$ on $\mSA$ and $\mSB$ as
$g \c \hQ^{(\A)} \equiv \hV_{g_\A} \hQ^{(\A)} \hV_{g_\A}^\dagger$ and
$g \c \hQ^{(\B)} \equiv \hW_{g_\B} \hQ^{(\B)} \hW_{g_\B}^\dagger$
for any $g = (g_\A, g_\B) \in \mG$.
A more complex example is given in Subsec.~\ref{subsec:sym_example}.

\subsection{Group covariant optimal solution}

We show that if Problem~P has a certain symmetry, then
there exists an optimal solution with the same type of symmetry
(proof in Appendix~\ref{append:sym}).

\begin{thm} \label{thm:sym}
 Suppose that, in Problem~P, there exist a group $\mG$ and its actions on $\mI_M$, $\mI_J$, and $\mS$ such that
 \begin{eqnarray}
  g \c \ha_{j,m} &=& \ha_{g \c j, g \c m}, ~ \forall g \in \mG, j \in \mI_J, m \in \mI_M, \nonumber \\
  b_j &=& b_{g \c j}, ~~~~~~ \forall g \in \mG, j \in \mI_J, \nonumber \\
  g \c \hc_m &=& \hc_{g \c m}, ~~ \forall g \in \mG, m \in \mI_M. \label{eq:sym_eval}
 \end{eqnarray}
 Then, as long as $\POVMA^\c$ is not empty, for any $\hPhi \in \POVMA^\c$,
 there exists $\hA \in \POVMA^\c$ such that $f(\hA) = f(\hPhi)$ and
 \begin{eqnarray}
  g \c \hAw &=& \hA(g \c \omega), ~~ \forall g \in \mG, \omega \in \Omega. \label{eq:sym_A}
 \end{eqnarray}
 Moreover, for any $(\hY, \nu) \in \mX^\c$,
 there exists $(\hX, \lambda) \in \mX^\c$ such that $s(\hX, \lambda) = s(\hY, \nu)$ and
 \begin{eqnarray}
  g \c \hX &=& \hX, ~~~~ \forall g \in \mG, \nonumber \\
  \lambda_j &=& \lambda_{g \c j}, ~~ \forall g \in \mG, j \in \mI_J. \label{eq:sym_Xlambda}
 \end{eqnarray}
 In particular, there exist an optimal solution $\hA$ to Problem~P satisfying Eq.~\eqref{eq:sym_A}
 and an optimal solution $(\hX, \lambda)$ to Problem~DP satisfying Eq.~\eqref{eq:sym_Xlambda}.
\end{thm}

If Eq.~\eqref{eq:sym_A} holds, then $\hPiA$ has the following symmetry:
\begin{eqnarray}
 g \c \hPiA_m &=& \hPiA_{g \c m}. \label{eq:sym_Pi}
\end{eqnarray}
Indeed, from Eqs.~\eqref{eq:sym_B}, \eqref{eq:sym_form1}, and \eqref{eq:sym_form2},
we obtain
\begin{eqnarray}
 g \c \hPiA_m &=& g \c \left[ \int_{\Omega} \hAdw \otimes \hBw_m \right] \nonumber \\
 &=& \int_{\Omega} [g \c \hAdw] \otimes \left[ g \c \hBw_m \right] \nonumber \\
 &=& \int_{\Omega} \hA[d(g \c \omega)] \otimes \hB^{g \c \omega}_{g \c m}
  = \hPiA_{g \c m}.
\end{eqnarray}

Let $\POVM_{\A;\mG}^\c$ be the entire set of $\hA \in \POVMA^\c$ satisfying Eq.~\eqref{eq:sym_A}
and $\mX_\mG^\c$ be the entire set of $(\hX, \lambda) \in \mX^\c$ satisfying Eq.~\eqref{eq:sym_Xlambda}.
We can easily verify that $\POVM_{\A;\mG}^\c$ and $\mX_\mG^\c$ are convex.
Thus, Problems~P and DP remain in convex programming
even if we restrict the feasible sets to $\POVM_{\A;\mG}^\c$ and $\mX_\mG^\c$, respectively.

\subsection{Example} \label{subsec:sym_example}

As an example of a symmetric problem,
let us consider the problem of finding a minimum-error sequential measurement
for ternary quantum states $\{ \hrho_m = \frac{1}{3} \trho^{(\A)}_m \otimes \trho^{(\B)}_m \}_{m=0}^2$
with $\Tr~\trho^{(\A)}_m = \Tr~\trho^{(\B)}_m = 1$,
where $\{ \hrho_m \}$ have the following symmetry.
Let $\mG_\A \equiv \{ p_\A^k, p_\A^k q_\A \}_{k \in \mI_3}$
and $\mG_\B \equiv \{ p_\B^k, p_\B^k q_\B \}_{k \in \mI_2}$
be dihedral groups with $|\mG_\A| = 6$ and $|\mG_\B| = 4$.
$\mG_k$ $(k \in \{ \A, \B \})$ is generated by a rotation $p_k$ and a reflection $q_k$,
which have $p_kq_kp_k = q_k$.
We have $p_\A^3 = q_\A^2 = e_\A$ and $p_\B^2 = q_\B^2 = e_\B$,
where $e_\A$ and $e_\B$ are, respectively, the identity elements of $\mG_\A$ and $\mG_\B$.
We define actions of $\mG_\A$ on $\mSA$ and $\mG_\B$ on $\mSB$ as
\begin{eqnarray}
 g_\A \c \hQ^{(\A)} &\equiv& \hV_{g_\A} \hQ^{(\A)} \hV_{g_\A}^\dagger,
  ~ \forall g_\A \in \mG_\A, \hQ^{(\A)} \in \mSA, \nonumber \\
 g_\B \c \hQ^{(\B)} &\equiv& \hW_{g_\B} \hQ^{(\B)} \hW_{g_\B}^\dagger,
  ~ \forall g_\B \in \mG_\B, \hQ^{(\B)} \in \mSB,
\end{eqnarray}
where $\hV_{g_\A}$ and $\hW_{g_\B}$ are, respectively,
unitary (or anti-unitary) operators on $\mHA$ and $\mHB$,
satisfying $\hV_{p_\A}^3 = \hV_{q_\A}^2 = \identA$ and $\hW_{p_\B}^2 = \hW_{q_\B}^2 = \identB$.
Assume
\begin{eqnarray}
 p_\A \c \trho^{(\A)}_m &=& \trho^{(\A)}_{m \oplus 1}, ~~
  q_\A \c \trho^{(\A)}_m = \trho^{(\A)}_{\kappa(m)}, \nonumber \\
 p_\B \c \trho^{(\B)}_m &=& \trho^{(\B)}_{\kappa(m)}, ~~
  q_\B \c \trho^{(\B)}_m = \trho^{(\B)}_m,
\end{eqnarray}
where $\kappa(0) = 0$, $\kappa(1) = 2$, and $\kappa(2) = 1$,
and $m \oplus 1$ is $m + 1$ if $m < 2$ holds; otherwise, 0.
For example, if $\{ \trho^{(\A)}_m \}_{m=0}^2$ are phase-shift keyed (PSK) coherent states
and $\{ \trho^{(\B)}_m \}_{m=0}^2$ are amplitude-shift keyed (ASK) coherent states,
then they have the above symmetries.
The phase space representation of such states is shown in Fig.~\ref{fig:symmetry}.
$p_\A$ and $p_\B$, respectively, correspond to the rotation of $2\pi/3$ and $\pi$.
$q_\A$ and $q_\B$ correspond to the reflection about the $x_{\rm c}$ axis.
\begin{figure}[tb]
 \centering
 \includegraphics[scale=0.8]{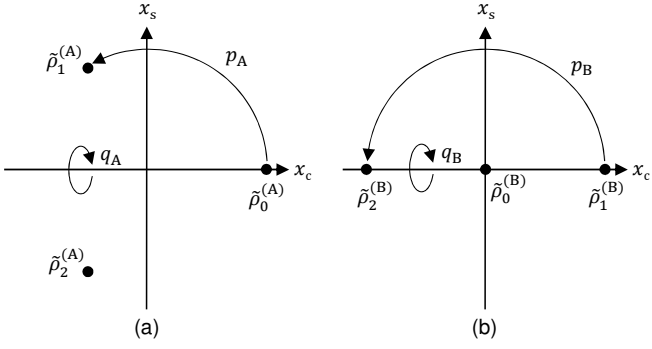}
 \caption{Phase-space representation of (a) PSK coherent states $\{ \trho^{(\A)}_m \}$
 and (b) ASK coherent states $\{ \trho^{(\B)}_m \}$.}
 \label{fig:symmetry}
\end{figure}

To use Theorem~\ref{thm:sym}, we obtain a group $\mG$ satisfying Eq.~\eqref{eq:sym_eval}.
Since $J = 0$ and $\hc_m = \hrho_m$ hold for the problem of finding a minimum-error sequential measurement,
Eq.~\eqref{eq:sym_eval} can be reduced to
\begin{eqnarray}
 g \c \hrho_m &=& \hrho_{g \c m}, ~~ \forall g \in \mG, m \in \mI_M. \label{eq:sym_ex}
\end{eqnarray}
Let $e \equiv (e_\A, e_\B)$ and $p \equiv (q_\A, p_\B)$;
then, $\mG_1 \equiv \{ e, p \}$ is a group such that
\begin{eqnarray}
 p \c \hrho_m &=& \frac{1}{3} \left[ q_\A \c \trho^{(\A)}_m \right] \otimes \left[ p_\B \c \trho^{(\B)}_m \right]
  \nonumber \\
 &=& \frac{1}{3} \trho^{(\A)}_{\kappa(m)} \otimes \trho^{(\B)}_{\kappa(m)}
  = \hrho_{\kappa(m)}.
\end{eqnarray}
Note that we redefine $\hV_g \equiv \hV_{g_\A}$ and $\hW_g \equiv \hW_{g_\B}$
for $g = (g_\A, g_\B) \in \mG_\A \times \mG_\B$.
Also, let $q \equiv (e_\A, q_\B)$;
then, $\mG_2 \equiv \{ e, q \}$ is a group such that
\begin{eqnarray}
 q \c \hrho_m &=& \frac{1}{3} \left[ e_\A \c \trho^{(\A)}_m \right] \otimes \left[ q_\B \c \trho^{(\B)}_m \right]
  \nonumber \\
 &=& \frac{1}{3} \trho^{(\A)}_m \otimes \trho^{(\B)}_m
  = \hrho_m.
\end{eqnarray}
$\mG_2$ expresses a symmetry of only $\mHB$.
We can consider the group $\mG = \{ e, p, q, pq \}$,
which is the direct product of $\mG_1$ and $\mG_2$.
Note that the action of $\mG$ on $\mSA$ is not faithful;
indeed, both $\hV_e$ and $\hV_q$ are identical to $\hV_{e_\A}$.
Let us define an action of $\mG$ on $\mI_M$ such that
$p \c m = \kappa(m)$ and $q \c m = m$;
then, Eq.~\eqref{eq:sym_ex} holds.
Thus, there exists $\hA \in \POVMA$ satisfying Eq.~\eqref{eq:sym_A}.
From Eq.~\eqref{eq:sym_Pi}, $\PiA$ with such $\hA$ has the following symmetry:
\begin{eqnarray}
 \hU_p \hPiA_m \hU_p^\dagger &=& \hPiA_{\kappa(m)}, \nonumber \\
 \hU_q \hPiA_m \hU_q^\dagger &=& \hPiA_m,
\end{eqnarray}
where $\hU_p = \hV_{q_\A} \otimes \hW_{p_\B}$ and
$\hU_q = \identA \otimes \hW_{q_\B}$ from Eq.~\eqref{eq:UVW}.
Moreover, from Eq.~\eqref{eq:sym_Xlambda},
there exists $\hX \in \mX^\c$ commuting with $\hV_{q_\A}$.

Note that, in this example, neither $\hA$ nor $\hX$ has the symmetry expressed by
$\{ p_\A^k \}_{k \in \mI_3}$,
while the states $\{ \trho^{(\A)}_m \}$ have this symmetry.
The reason is that the states $\{ \trho^{(\B)}_m \}$ do not have this symmetry.

\section{Generalized minimax solution} \label{sec:minimax}

In the minimax strategy for a quantum state discrimination problem,
prior probabilities is unknown and the task is to maximize the worst case of
the objective function (such as the average success probability)
over all prior probabilities.
This strategy has been investigated in several studies
\cite{Hir-Ike-1982,Osa-Ban-Hir-1996,Dar-Sac-Kah-2005,Kat-2012,Nak-Kat-Usu-2013-minimax},
whose generalized version is appeared in Ref.~\cite{Nak-Kat-Usu-2015-general}.
In this section, we consider a sequential-measurement version of
the generalized minimax problem.
In a similar manner to the method reported by Ref.~\cite{Nak-Kat-Usu-2015-general},
we can provide necessary and sufficient conditions for a minimax solution
to the sequential-measurement version of the problem.
In what follows, we discuss properties that a minimax solution has.

\subsection{Formulation}

Let us consider $K \ge 1$ objective functions $f_0(\hA), \cdots, f_{K-1}(\hA)$
expressed as:
\begin{eqnarray}
 f_k(\hA) &\equiv& \summ \Tr \left[ \hc_{k,m} \hPiA_m \right] + d_k, \label{eq:fmm_k}
\end{eqnarray}
where $\hc_{k,m} \in \mS$ and $d_k \in \Real$.
Also, let $\Prob$ be the entire set of collections of $K$ nonnegative real numbers,
$\mu \equiv \{ \mu_k \}_{k=0}^{K-1} \in \Realp^K$, satisfying $\sum_{k=0}^{K-1} \mu_k = 1$.
$\mu \in \Prob$ can be interpreted as a probability distribution.
Let $F(\mu, \hA)$ be the objective function defined by
\begin{eqnarray}
 F(\mu, \hA) &\equiv& \sum_{k=0}^{K-1} \mu_k f_k(\hA) \label{eq:fmm}
\end{eqnarray}
and $\POVMA^\c$ be the set defined by Eq.~\eqref{eq:POVMc}.
We investigate the problem of finding a POVM $\hA \in \POVMA^\c$ that maximizes
the worst-case value of $F(\mu, \hA)$ over $\mu \in \Prob$.
This problem can be formulated as follows:
\begin{eqnarray}
 \begin{array}{lll}
  {\rm P_m:} & {\rm maximize} & \displaystyle \min_{\mu \in \Prob} F(\mu, \hA) \\
  & {\rm subject~to} & \hA \in \POVMA^\c. \\
 \end{array} \label{eq:primal_mm}
\end{eqnarray}
Let $F^\opt$ be the optimal value of Problem~$\Pm$.
We call $(\mu^\opt, \hA^\opt) \in \Prob \times \POVMA^\c$ and $\hA^\opt \in \POVMA^\c$
satisfying $F(\mu^\opt, \hA^\opt) = F^\opt$
a minimax solution and a minimax POVM, respectively.

\subsection{Properties of a minimax solution}

We first show the following remark.
\begin{remark}[Minimax theorem] \label{remark:minimax}
 If $\POVMA^\c$ is not empty, then
 there exsists a minimax solution $(\mu^\opt, \hA^\opt)$ to Problem~$\Pm$,
 and it satisfies
 \begin{eqnarray}
  \min_{\mu \in \Prob} \max_{\hA \in \POVMA^\c} F(\mu, \hA) &=& F(\mu^\opt, \hA^\opt) \nonumber \\
  &=& \max_{\hA \in \POVMA^\c} \min_{\mu \in \Prob} F(\mu, \hA). \label{eq:minimax}
 \end{eqnarray}
\end{remark}

\begin{proof}
 $\POVMA^\c$ and $\Prob$ are closed convex sets.
 $F(\mu, \hA)$ is a continuous convex function of $\mu$ for fixed $\hA$
 and a continuous concave function of $\hA$ for fixed $\mu$.
 Thus, the minimax theorem holds (e.g., \cite{Eke-Tem-1999});
 that is to say, there exists a minimax solution $(\mu^\opt, \hA^\opt)$ to Problem~$\Pm$,
 which satisfies Eq.~\eqref{eq:minimax}.
 \QED
\end{proof}

A minimax solution to Problem~$\Pm$ can be characterized by a saddle point;
i.e., $(\mu^\opt, \hA^\opt)$ is a minimax solution if and only if,
for any $\mu \in \Prob$ and $\hA \in \POVMA^\c$, $(\mu^\opt, \hA^\opt)$ satisfies \cite{Eke-Tem-1999}
\begin{eqnarray}
 F(\mu^\opt, \hA) \le F(\mu^\opt, \hA^\opt) \le F(\mu, \hA^\opt). \label{eq:minimax_saddle}
\end{eqnarray}

Let
\begin{eqnarray}
 F^\opt(\mu) &\equiv& \max_{\hA \in \POVMA^\c} F(\mu, \hA). \label{eq:Fstar}
\end{eqnarray}
From Eq.~\eqref{eq:minimax_saddle}, $F^\opt(\mu^\opt) = F(\mu^\opt, \hA^\opt)$ holds.

Let $\ol{c_m}(\mu) \equiv \sum_{k=0}^{K-1} \mu_k \hc_{k,m}$ and
$\ol{d}(\mu) \equiv \sum_{k=0}^{K-1} \mu_k d_k$;
then, we find that the problem of finding $F^\opt(\mu)$ for a fixed $\mu \in \Prob$
is reduced to Problem~P, as shown in the following remark:
\begin{remark}
 Let $\ol{f^\opt}(\mu)$ be the optimal value of Problem~P with $\hc_m = \ol{c_m}(\mu)$;
 then, $F^\opt(\mu) = \ol{f^\opt}(\mu) + \ol{d}(\mu)$ holds.
\end{remark}

\begin{proof}
 \begin{eqnarray} \hspace{-1.5em}
  F^\opt(\mu) &=& \max_{\hA \in \POVMA^\c} F(\mu, \hA) \nonumber \\
  &=& \max_{\hA \in \POVMA^\c} \sum_{k=0}^{K-1} \mu_k \left[ \summ \Tr \left[ \hc_{k,m} \hPiA_m \right] + d_k \right]
   \nonumber \\
  &=& \max_{\hA \in \POVMA^\c} \summ \Tr\left[ \ol{c_m}(\mu) \hPiA_m \right] + \ol{d}(\mu)
   \nonumber \\
  &=& \ol{f^\opt}(\mu) + \ol{d}(\mu).
 \end{eqnarray}
 \QED
\end{proof}

\begin{thm} \label{thm:minimax}
 Assume $\mu^\opt \in \Prob$ and $\hA^\opt \in \POVMA^\c$.
 The following statements are all equivalent.
 \begin{enumerate}[(1)]
  \setlength{\parskip}{0cm}
  \setlength{\itemsep}{0cm}
  \item $(\mu^\opt, \hA^\opt)$ is a minimax solution to Problem~$\Pm$.
  \item The following equation holds:
		\begin{eqnarray}
		 f_k(\hA^\opt) &\ge& F^\opt(\mu^\opt), ~~ \forall k \in \mI_K. \label{eq:thm_minimax2}
		\end{eqnarray}
  \item The following equations hold:
		\begin{eqnarray}
         F^\opt(\mu^\opt) &=& F(\mu^\opt, \hA^\opt), \nonumber \\
		 f_k(\hA^\opt) &\ge& f_{k'}(\hA^\opt), ~~ \forall k,k' \in \mI_K ~{\rm s.t.}~\mu^\opt_{k'} > 0.
		  \label{eq:thm_minimax3}
		\end{eqnarray}
 \end{enumerate}
\end{thm}

\begin{proof}
 The same as Theorem~3 of Ref.~\cite{Nak-Kat-Usu-2015-general}.
 \QED
\end{proof}

\begin{thm} \label{thm:minimax_convex}
 Let us consider the following optimization problem
 \begin{eqnarray}
  \begin{array}{ll}
   {\rm maximize} & \displaystyle f_{\rm min}(\hA) \equiv \min_{k \in \mI_K} f_k(\hA) \\
   {\rm subject~to} & \hA \in \POVMA^\c \\
  \end{array} \label{eq:minimax_convex}
 \end{eqnarray}
 with $\hA$.
 An optimal solution to the problem given by Eq.~\eqref{eq:minimax_convex}
 is equivalent to a minimax POVM of Problem~$\Pm$.
\end{thm}

\begin{proof}
 The same as Theorem~4 of Ref.~\cite{Nak-Kat-Usu-2015-general}.
 \QED
\end{proof}

\subsection{Group covariant minimax problem}

Similar to Theorem~\ref{thm:sym},
if Problem~$\Pm$ has a certain symmetry, then
there exists a minimax solution with the same type of symmetry,
as stated in the following theorem
(proof in Appendix~\ref{append:sym_minimax}).

\begin{thm} \label{thm:sym_minimax}
 Suppose that, in Problem~$\Pm$,
 there exist a group $\mG$ and its actions on $\mI_M$, $\mI_J$, $\mI_K$, and $\mS$ such that
 \begin{eqnarray}
  g \c \ha_{j,m} &=& \ha_{g \c j, g \c m}, ~ \forall g \in \mG, j \in \mI_J, m \in \mI_M, \nonumber \\
  b_j &=& b_{g \c j}, ~~~~~~ \forall g \in \mG, j \in \mI_J, \nonumber \\
  g \c \hc_{k,m} &=& \hc_{g \c k, g \c m}, ~~ \forall g \in \mG, k \in \mI_K, m \in \mI_M, \nonumber \\
  d_k &=& d_{g \c k}, ~~ \forall g \in \mG, k \in \mI_K. \label{eq:sym_eval_mm}
 \end{eqnarray}
 Then, as long as $\POVMA^\c$ is not empty,
 there exists a minimax solution $(\mu^\opt, \hA^\opt)$ such that
 \begin{eqnarray}
  \mu^\opt_k &=& \mu^\opt_{g \c k}, ~~~ \forall g \in \mG, k \in \mI_K, \nonumber \\
  g \c \hA^\opt(\omega) &=& \hA^\opt(g \c \omega), ~~ \forall g \in \mG, \omega \in \Omega. \label{eq:sym_minimax}
 \end{eqnarray}
\end{thm}

\section{Examples} \label{sec:example}

In this section, we apply our results to the problem of finding an optimal inconclusive sequential measurement.
Also, we derive an analytical expression of an optimal inconclusive sequential measurement for double trine states.
Note that one can find other examples of generalized state discrimination problems in Subsec.~II.B
of Ref.~\cite{Nak-Kat-Usu-2015-general}.

\subsection{Optimal inconclusive sequential measurement} \label{subsec:inc}

An optimal inconclusive measurement is a measurement that maximizes the average success probability
with a fixed average inconclusive probability, $\pinc$.
We here consider its sequential-measurement version.

Let us consider the problem of obtaining an optimal inconclusive sequential measurement,
$\PiA = \{ \hPiA_r \}_{r=0}^R$ $(R \ge 2)$,
for the states $\{ \trho_r \}_{r=0}^{R-1}$ with prior probabilities $\{ \xi_r \}_{r=0}^{R-1}$.
The detection operator $\hPiA_r$ $~(r \in \mI_R)$ corresponds to identification of
the state $\hrho_r$, while $\hPiA_R$ corresponds to the inconclusive answer.
The problem can be formulated as follows:
\begin{eqnarray}
 \begin{array}{ll}
  {\rm maximize} & \displaystyle \PS(\hA) \equiv \sumr \Tr \left[ \hrho_r \hPiA_r \right] \\
  {\rm subject~to} & \displaystyle \hA \in \POVMA, ~ \sumr \Tr \left[ \hrho_r \hPiA_R \right] = \pinc, \\
 \end{array} \label{eq:inc_primal}
\end{eqnarray}
where $\hrho_r \equiv \xi_r \trho_r$.
This problem is equivalent to Problem~P with
\begin{eqnarray}
 M &=& R + 1, \nonumber \\
 J &=& 1, \nonumber \\
 \hc_m &=&
  \left\{
   \begin{array}{ll}
	\hrho_m, & ~ m < R, \\
	0, & ~ m = R,
   \end{array} \right. \nonumber \\
 \ha_{0,m} &=&
  \left\{
   \begin{array}{ll}
	0, & ~ m < R, \\
	\displaystyle - \sumr \hrho_r, & ~ m = R,
   \end{array} \right. \nonumber \\
 b_0 &=& - \pinc, \label{eq:inc_param}
\end{eqnarray}
where we use the fact that
the problem remains unchanged when the second constraint of Eq.~\eqref{eq:inc_primal}
is replaced with $\sumr \Tr[\hrho_r \hPiA_R] \ge \pinc$.
Substituting Eq.~\eqref{eq:inc_param} into Problem~DP yields
the following dual problem:
\begin{eqnarray}
 \begin{array}{ll}
  {\rm minimize} & \displaystyle s(\hX, \lambda) = \Tr~\hX - \lambda \pinc \\
  {\rm subject~to} & (\hX, \lambda) \in \mX^\c, \\
 \end{array} \label{eq:inc_dual}
\end{eqnarray}
where
\begin{eqnarray}
 \mX^\c &=& \left\{ (\hX, \lambda) \in \mSA \otimes \Realp : \hX \ge \hsw(\lambda),
              ~ \forall \omega \in \Omega \right\}, \nonumber \\
 \hsw(\lambda) &=& \TrB \sumr \hrho_r \left[ \hBw_r + \lambda \hBw_R \right].
  \label{eq:inc_hsw}
\end{eqnarray}

From Theorem~\ref{thm:condition},
$\hA \in \POVMA^\c$ is an optimal solution if and only if
the following equations hold:
\begin{eqnarray}
 [\hX^\opt - \hsw(\lambda^\opt)] \hAw &=& 0, ~ \forall \omega \in \Omega, \nonumber \\
 \lambda^\opt \left[ \sumr \Tr \left[ \hrho_r \hPiA_R \right] - \pinc \right] &=& 0,
  \label{eq:inc_cond}
\end{eqnarray}
where $(\hX^\opt, \lambda^\opt)$ is an optimal solution to Eq.~\eqref{eq:inc_dual}.

\subsection{Optimal inconclusive sequential measurement for double trine states} \label{subsec:dtrine}

We derive an optimal solution to the problem of Eq.~\eqref{eq:inc_primal}
for double trine states with equal probabilities.
Note that, in the cases of $\pinc = 0$ (corresponding to a minimum-error sequential measurement)
and $\pinc = 1/2$ (corresponding to an optimal unambiguous sequential measurement),
optimal solutions are given in Refs.~\cite{Chi-Hsi-2013} and \cite{Chi-Dua-Hsi-2014},
respectively.

Double trine states with equal probabilities can be expressed by
$\{ \hrho_m \equiv \frac{1}{3} \ket{\psi_m}\bra{\psi_m} \}_{m=0}^2$ with
\begin{eqnarray}
 \ket{\psi_m} &\equiv& \ket{\phi_m} \otimes \ket{\phi_m}, \nonumber \\
 \ket{\phi_m} &\equiv& \cos \frac{2\pi m}{3} \ket{0} + \sin \frac{2\pi m}{3} \ket{1}. \label{eq:dtrine_phi}
\end{eqnarray}
$\{ \ket{\phi_m} \}$ has the symmetry of $\ket{\phi_m} = \hV_{\rm rot}^k \ket{\phi_{m \ominus k}}$,
where
\begin{eqnarray}
 \hV_{\rm rot} &\equiv& - \frac{1}{2} \ident + \frac{\sqrt{3}}{2} (\ket{1}\bra{0} - \ket{0}\bra{1}),
\end{eqnarray}
which is a unitary operator corresponding to a rotation of $\frac{2\pi}{3}$,
and $m \ominus k$ is the remainder of the division of $m - k$ by 3.
Also, since $\braket{k | \phi_m}$ $~(k \in \{ 0, 1 \})$ is real,
$\hV_{\rm conj} \ket{\phi_m} = \ket{\phi_m}$ holds,
where $\hV_{\rm conj}$ is the anti-unitary operator of complex conjugation
in the basis $\{ \ket{0}, \ket{1} \}$
\footnote{
Our discussion in Sec.~\ref{sec:sym} can be used 
when considering a dihedral group with order 6, $\mG = \{ p^k, p^k q \}_{k \in \mI_3}$,
which is generated by a rotation $p$ and a reflection $q$ with $pqp = q$.
To be concrete, let $\hV_{p^kq^l} = \hV_{\rm rot}^k \hV_{\rm conj}^l$
for any $k \in \mI_3$ and $l \in \mI_2$,
and let $\hU_g = \hV_g \otimes \hV_g$;
then, we can consider group actions of $\mG$.
Note that double trine states also have the symmetry of
$(\ket{0}\bra{0} - \ket{1}\bra{1})\ket{\phi_m} = \ket{\phi_{\kappa(m)}}$
($\kappa(0) = 0$, $\kappa(1) = 2$, and $\kappa(2) = 1$);
however, we do not need this symmetry to obtain their optimal sequential measurement.}.

First, we derive an optimal solution $(\hX^\opt, \lambda^\opt)$ to the problem
of Eq.~\eqref{eq:inc_dual}.
Assume, without loss of generality, that $\hX^\opt$ commutes with $\hV_{\rm rot}$ and $\hV_{\rm conj}$
(see Theorem~\ref{thm:sym});
then, it follows that such $\hX^\opt$ must be proportional to $\identA$.
After some computations, we obtain an optimal solution $(\hX^\opt, \lambda^\opt)$ as follows
(see Appendix~\ref{append:dtrine}):
\begin{eqnarray}
 \hX^\opt &=& \left( \frac{1}{2} + \frac{3 - 2\pinc}{4 \sqrt{3 - 4 \pinc}} \right) \identA, \nonumber \\
 \lambda^\opt &=& \frac{1}{2} + \frac{1}{2\sqrt{3 - 4 \pinc}}.
  \label{eq:dtrine_lambda_opt}
\end{eqnarray}
Thus, the average success probability of an optimal inconclusive sequential measurement, $\PS^\opt$,
which is equivalent to the optimal value $s(\hX^\opt, \lambda^\opt)$, is given by
\begin{eqnarray}
 \PS^\opt &=& s(\hX^\opt, \lambda^\opt) = \Tr~\hX^\opt - \lambda^\opt \pinc \nonumber \\
  &=& \frac{1}{2}(1 - \pinc)  + \frac{1}{4} \sqrt{3 - 4 \pinc}.
\end{eqnarray}
When $\pinc = 1/2$, $\PS^\opt + \pinc = 1$ holds;
i.e., the average error probability, $1 - \PS^\opt - \pinc$, is zero.
This indicates that there exists an unambiguous sequential measurement
with the average inconclusive probability of $1/2$.
Since the case of $\pinc > 1/2$ is trivial, assume $0 \le \pinc \le 1/2$
(in this case, $\frac{1}{2} + \frac{1}{2\sqrt{3}} \le \lambda^\opt \le 1$ holds).

Next, we derive an optimal sequential measurement.
Let $\ket{\phi^\perp_m}$ be the vector expressed by
\begin{eqnarray}
 \ket{\phi^\perp_m} &\equiv& - \sin \frac{2\pi m}{3} \ket{0} + \cos \frac{2\pi m}{3} \ket{1},
\end{eqnarray}
which satisfies $\braket{\phi^\perp_m | \phi_m} = 0$ and
$\ket{\phi^\perp_m} = \hV_{\rm rot}^k \ket{\phi^\perp_{m \ominus k}}$.
From the discussion in Appendix~\ref{append:dtrine} and the symmetry of $\{ \ket{\phi_m} \}$,
$\hX^\opt - \hsw(\lambda^\opt)$ is rank one
(i.e., the largest eigenvalue of $\hsw(\lambda^\opt)$ is $\upsilon(\lambda^\opt)$,
 which is defined in Appendix~\ref{append:dtrine})
if and only if $\{ \hBw_m \}_{m=0}^3$ is expressed as
\begin{eqnarray}
 \hBw_m &=& \hB^{(\omega_k)}_m \equiv
  \left\{
   \begin{array}{ll}
    \hV_{\rm rot}^k \hB^\bullet_{m \ominus k} \hV_{\rm rot}^{-k}, & m < 3, \\
    \displaystyle \frac{4}{3} \pinc \ket{\phi^\perp_k}\bra{\phi^\perp_k}, & m = 3, \\
   \end{array} \right. \label{eq:dtrine_hBw}
\end{eqnarray}
where $\{ \hB^\bullet_m \}$ is given by Eq.~\eqref{eq:dtrine_B} with $\alpha = 4\pinc/3$,
and $\omega_k \in \Omega$ $~(k \in \mI_3)$ is an index
corresponding to the POVM $B^{(\omega_k)} \equiv \{ \hB^{(\omega_k)}_m \}$
defined by Eq.~\eqref{eq:dtrine_hBw}.
In Eq.~\eqref{eq:dtrine_hBw}, we use
\begin{eqnarray}
 \hV_{\rm rot}^k \hB^\bullet_3 \hV_{\rm rot}^{-k}
  &=& \alpha \hV_{\rm rot}^k \ket{\phi^\perp_0} \bra{\phi^\perp_0} \hV_{\rm rot}^{-k}
 = \frac{4}{3} \pinc \ket{\phi^\perp_k}\bra{\phi^\perp_k}. \nonumber \\
\end{eqnarray}
Using the fact that, from Eq.~\eqref{eq:inc_cond},
the support of $\hA(\omega)$ is included in the kernel of $\hX^\opt - \hsw(\lambda^\opt)$,
$\hA$ can be obtained in the following way.
When $\omega = \omega_k$,
since Eq.~\eqref{eq:dtrine_schatten} with $\theta = \frac{2\pi k}{3}$ holds,
$\hV_{\rm rot}^k \ket{1} = \ket{\phi^\perp_k}$ is in the kernel of $\hX^\opt - \hsw(\lambda^\opt)$.
Then, $\hA(\omega_k)$ must be proportional to $\ket{\phi^\perp_m} \bra{\phi^\perp_m}$.
When $\omega \neq \omega_k$, the kernel of $\hX^\opt - \hsw(\lambda^\opt)$ is $\{ 0 \}$,
this implies that $\hA(\omega) = 0$.
Thus, an optimal inconclusive sequential measurement, $\Pi^\opt$,
is expressed by $\Pi^\opt = \PiA$, where
\begin{eqnarray}
 \hA(\omega) &=&
  \left\{
   \begin{array}{ll}
    \displaystyle \frac{2}{3} \ket{\phi^\perp_k} \bra{\phi^\perp_k},
     & \omega = \omega_k ~ (k \in \mI_3), \\
    0, & {\rm otherwise} \\
   \end{array} \right. \label{eq:dtrine_A}
\end{eqnarray}
holds for any $\omega \in \Omega$.
It follows that $\hA$ is a POVM with three outcomes, $\{ \omega_k \}_{k=0}^2$.
From Eqs.~\eqref{eq:dtrine_hBw} and \eqref{eq:dtrine_A},
$\Pi^\opt$ can be rewritten as
\begin{eqnarray}
 \hPi^\opt_m &=& 
  \left\{
   \begin{array}{ll}
	\displaystyle \frac{2}{3} \sum_{\substack{k=0 \\ k \neq m}}^2
     \ket{\phi^\perp_k}\bra{\phi^\perp_k} \otimes
     \hV_{\rm rot}^k \hB^\bullet_{m \ominus k} \hV_{\rm rot}^{-k}, & ~ m < 3, \\
	\displaystyle \frac{8}{9} \pinc \sum_{k=0}^2
     \ket{\phi^\perp_k}\bra{\phi^\perp_k} \otimes
     \ket{\phi^\perp_k}\bra{\phi^\perp_k}, & ~ m = 3. \\
   \end{array} \right.
\end{eqnarray}

Figure~\ref{fig:sequential} shows the average success probabilities
of optimal measurements with and without the restriction that
only sequential measurements are allowed.
Note that the average success probability of an optimal inconclusive global measurement
can be computed by the method described in Ref.~\cite{Nak-Usu-Kat-2012-GUInc}.
The average error probability is zero when $\pinc \ge 1/2$ and $\pinc \ge 1/4$
in the cases of optimal inconclusive sequential and global measurements, respectively.
\begin{figure}[tb]
 \centering
 \includegraphics[scale=0.6]{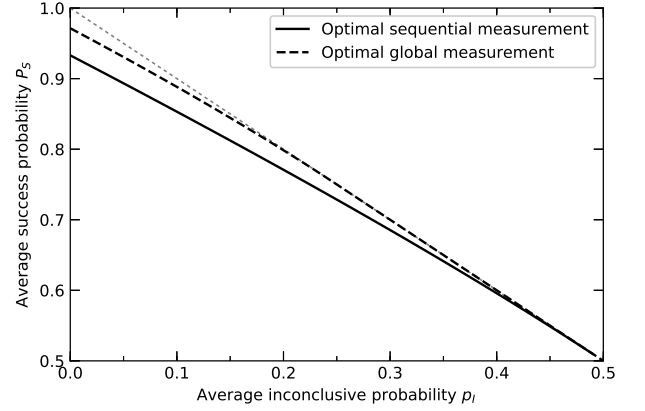}
 \caption{Average success probabilities of optimal sequential and global measurements for double trine states
 with equal prior probabilities.}
 \label{fig:sequential}
\end{figure}

\section{Conclusion}

We have studied a sequential-measurement version of the generalized state discrimination problem
discussed in Ref.~\cite{Nak-Kat-Usu-2015-general}.
Since the entire set of sequential measurements is convex, Problem~P is convex programming.
The corresponding dual problem and necessary and sufficient conditions
for an optimal sequential measurement were derived.
We also showed that for an optimization problem having a certain group symmetry,
there exists an optimal solution with the same type of symmetry.
Moreover, the minimax version of this problem was studied,
and necessary and sufficient conditions for a minimax solution were provided.
We expect that our results will be useful for the investigation
of a broad class of state discrimination problems with sequential measurements.

\begin{acknowledgments}
 We are grateful to O. Hirota of Tamagawa University for support.
 T. S. U. was supported (in part) by JSPS KAKENHI (Grant No.16H04367).
\end{acknowledgments}

\appendix

\section{Proof of Theorem~\ref{thm:zero_gap}} \label{append:zero_gap}

We will prove the cases of $f^\opt > -\infty$ and $f^\opt = -\infty$ separately.

\subsection{Case of $f^\opt > -\infty$}

From $s^\opt \ge f^\opt$, it is sufficient to show that there exists $\hA \in \POVMA^\c$
satisfying $f(\hA) \ge s^\opt$.
Indeed, in this case, $s^\opt = f^\opt$ holds from $s^\opt \le f(\hA) \le f^\opt$.

Let us consider the following set:
\begin{eqnarray}
 \mZ &\equiv& \Bigg\{ \left( \{ \hsw(\lambda) + \hxw - \hX \}_{\omega \in \Omega},
                       s(\hX, \lambda) - u \right) \nonumber \\
 & & ~~~ : (\hX, \lambda, u, \{\hxw\}_{\omega \in \Omega}) \in \mT \Bigg\},
\end{eqnarray}
where
\begin{eqnarray}
 \mT &\equiv& \Bigg\{ (\hX, \lambda, u, \{\hxw\}) : 
  (\hX, \lambda) \in \mX, s^\opt > u \in \Real, \nonumber \\
 & & ~~~~~ \hxw \in \mSA^+ \Bigg\}.
\end{eqnarray}
Since $\hxw$ is in $\mSA^+$, $\hsw(\lambda) + \hxw - \hX = 0$ holds only if $\hX \ge \hsw(\lambda)$ holds,
which implies that $\{ \hsw(\lambda) + \hxw - \hX \}_\omega = \{0\}$ holds
only if $(\hX, \lambda) \in \mX^\c$ holds.
Since $s(\hX, \lambda) \ge s^\opt > u$ holds when $(\hX, \lambda) \in \mX^\c$,
we have $(\{0\}, 0) \not\in \mZ$.
Also, we can easily see that $\mZ$ is a convex set having a nonempty interior.
Thus, from the geometric Hahn-Banach theorem (e.g., \cite{Lue-1969}),
for any $(\hX, \lambda, u, \{\hxw\}) \in \mT$,
there exists $(\{ \tAw \}_{\omega \in \Omega}, \alpha) \neq (\{0\},0)$ with $\tAw \in \mSA$ and $\alpha \in \Real$
satisfying
\begin{eqnarray}
 \Tr \int_\Omega \tAw [\hsw(\lambda) + \hxw - \hX] \mu(\dw)
  + \alpha [s(\hX, \lambda) - u] \ge 0, \nonumber \\
 \label{eq:DPL_separation0}
\end{eqnarray}
where $\mu$ is a strictly positive measure on a sigma algebra $\sigma(\Omega)$ satisfying $\mu(\Omega) = 1$.
Let $\delta_{\omega}(E)$ $~(E \subseteq \Omega)$ be the Dirac measure,
which is defined by $\delta_{\omega}(E) = 1$ if $\omega \in E$ holds,
$\delta_{\omega}(E) = 0$ otherwise.
By substituting $\hxw = t \hx \delta_{\omega'}(\omega)$ $~(\hx \ge 0, \omega' \in \Omega)$
into Eq.~\eqref{eq:DPL_separation0} and taking the limit $t \to \infty$,
we obtain $\Tr(\tA_{\omega'}\hx) \ge 0$.
Since this inequality holds for any $\hx \ge 0$ and $\omega' \in \Omega$,
$\tA_\omega \ge 0$ holds for any $\omega \in \Omega$.
Also, taking the limit $u \to -\infty$ in Eq.~\eqref{eq:DPL_separation0} gives
$\alpha \ge 0$.

To show $\alpha > 0$, assume by contradiction that $\alpha = 0$.
Substituting 
$\hX = t\hx$ $~(\hx \ge 0)$ and $\hxw = t[1 - \delta_{\omega'}(\omega)]\hx$
into Eq.~\eqref{eq:DPL_separation0} and taking the limit $t \to \infty$
gives
\begin{eqnarray}
 \Tr \int_\Omega \tAw \delta_{\omega'}(\omega) \hx \mu(\dw) &=& \Tr(\tA_{\omega'} \hx) \le 0,
\end{eqnarray}
which implies $\tA_\omega \le 0$ for any $\omega \in \Omega$.
Thus, $\tA_\omega = 0$ must hold, which contradicts $(\{ \tAw \}, \alpha) \neq (\{0\},0)$.
Therefore, $\alpha > 0$ holds.

Here, let $\hA \in \oPOVMA$ be a measure satisfying $\hAw = \tAw \mu(\omega) / \alpha$
for any $\omega \in \Omega$.
To complete the proof, we will show $\hA \in \POVMA^\c$ and $f(\hA) \ge s^\opt$.
Dividing both sides of Eq.~\eqref{eq:DPL_separation0} by $\alpha$ yields
\begin{eqnarray}
 \hspace{-1.5em}
 \Tr \int_\Omega [\hsw(\lambda) + \hxw - \hX] \hAdw
  + s(\hX, \lambda) - u \ge 0. \label{eq:DPL_separation}
\end{eqnarray}
Substituting $\hX = t\hx$ $~(\hx \in \mSA)$ into Eq.~\eqref{eq:DPL_separation}
and taking the limit $t \to \infty$ gives
\begin{eqnarray}
 \Tr~\hx &\ge& \Tr \left[ \hx \int_\Omega \hAdw \right]
  = \Tr [\hx \hA(\Omega)].
\end{eqnarray}
Since this inequality holds for any $\hx \in \mSA$,
$\hA(\Omega) = \identA$ holds.
Substituting $\lambda_j = t \delta_{j,j'}$ ($\delta_{j,j'}$ is Kronecker delta)
into Eq.~\eqref{eq:DPL_separation}
and taking the limit $t \to \infty$ gives $\eta_{j'}(\hA) \le 0$,
and thus $\hA \in \POVMA^\c$ holds.
Also, substituting $\hxw = \hX = 0$ and $\lambda = 0$ into Eq.~\eqref{eq:DPL_separation}
and taking the limit $u \to s^\opt$ gives $f(\hA) \ge s^\opt$.
Therefore, $s^\opt = f^\opt$ holds.

\subsection{Case of $f^\opt = -\infty$}

Let us consider the following set:
\begin{eqnarray}
 \mW &=& \left\{ \{ \eta_j(\hA) \}_{j=0}^{J-1} \in \Real^J : \hA \in \POVMA \right\}.
\end{eqnarray}
Since $f^\opt = -\infty$ implies that $\POVMA^\c$ is empty,
for any $\hA \in \POVMA$, there exists $j \in \mI_J$ such that $\eta_j(\hA) > 0$.
Therefore, the set $\mW' \equiv \{ \{ \beta_j \le 0 \}_{j=0}^{J-1} \in \Real^J \}$
has no intersecton with $\mW$.
We can easily verify that $\mW$ is compact and $\mW'$ is closed;
thus, by a separating hyperplane theorem (e.g., \cite{Dha-Dut-2011}),
there exist $q \equiv \{ q_j \}_{j=0}^{J-1} \in \Realp^J$ and $0 < \epsilon \in \Realp$
such that
\begin{eqnarray}
 \sumj q_j \eta_j(\hA) &>& \epsilon, ~~ \forall \hA \in \POVMA. \label{eq:zerogap_qeta}
\end{eqnarray}
Now, assume that $\sumj q_j = 1$, with no loss of generality.
Equations~\eqref{eq:POVMc} and \eqref{eq:zerogap_qeta} give
\begin{eqnarray}
 \sumj q_j \Tr \summ \ha_{j,m} \hPiA &\ge& \sumj q_j b_j + \epsilon.
\end{eqnarray}
Substituting Eq.~\eqref{eq:PiA} into this equation and doing some algebra gives
\begin{eqnarray}
 \Tr \int_\Omega \Xi(\omega) \hA(\dw) \ge 0, \label{eq:zerogap_XiA}
\end{eqnarray}
where
\begin{eqnarray}
 \Xi(\omega) &\equiv& \TrB \sumj q_j \summ \ha_{j,m} \hBw_m \nonumber \\
 & & \mbox{} - \left( \sumj q_j b_j + \epsilon \right) \frac{\identA}{\dA}
\end{eqnarray}
and $\dA = \dim~\mHA$.
Since Eq.~\eqref{eq:zerogap_XiA} holds for any $\hA \in \POVMA$,
$\Xi(\omega) \ge 0$ holds.

Let $\hX^\opt_0$ be the optimal solution to the following problem:
\begin{eqnarray}
 \begin{array}{ll}
  {\rm minimize} & \Tr~\hX_0 \\
  {\rm subject~to} & \displaystyle \hX_0 \ge \TrB \summ \hc_m \hBw_m, ~~ \forall \omega \in \Omega. \\
 \end{array} \label{eq:zerogap_X0}
\end{eqnarray}
Also, let
\begin{eqnarray}
 \hY(t, q) &\equiv& \hX^\opt_0 - t \left( \sumj q_j b_j + \epsilon \right) \frac{\identA}{\dA},
  \label{eq:zerogap_Y}
\end{eqnarray}
where $t \in \Realp$.
From Eqs.~\eqref{eq:zerogap_X0} and \eqref{eq:zerogap_Y}, we have
\begin{eqnarray}
 \hY(t, q) &\ge& \TrB \summ \hc_m \hBw_m - t \TrB \sumj q_j \summ \ha_{j,m} \hBw_m \nonumber \\
 &=& \TrB \summ \hz_m(tq) \hBw_m = \hsw(tq),
\end{eqnarray}
where the first and second lines follow from $\Xi(\omega) \ge 0$
and the definition of $\hz_m(\lambda)$ given by Eq.~\eqref{eq:z}, respectively.
Thus, $[\hY(t,q), tq] \in \mX^\c$ holds,
which gives $s[\hY(t,q), tq] \ge s^\opt$.
From Eq.~\eqref{eq:zerogap_Y}, we obtain
\begin{eqnarray}
 s^\opt &\le& s[\hY(t,q), tq] = \Tr~\hY(t,q) + t \sumj q_j b_j \nonumber \\
 &=& \Tr~\hX^\opt_0 - t \epsilon.
\end{eqnarray}
Since $\Tr~\hX^\opt_0$ is constant, $\Tr~\hX^\opt_0 - t \epsilon \to -\infty$ as $t \to \infty$.
Therefore, $s^\opt = -\infty$ holds.
\QED

\section{Proof of Theorem~\ref{thm:finite}} \label{append:finite}

\subsection{Outline}

Let $b \equiv \{ b_j \}_{j=0}^{J-1} \in \Real^J$.
Also, let $f^\bullet(\beta)$ be the optimal value of the optimization problem
obtained by replacing $b$ of Problem~P with $\beta \equiv \{ \beta_j \}_{j=0}^{J-1} \in \Real^J$.
We will first show that $f^\bullet(\beta)$ is a concave function.
We will also show that there exists an optimal solution to Problem~P with at most $\dA^2$ outcomes
if $f^\bullet(\beta)$ is strictly concave at $\beta = b$,
and with at most $(J + 1) \dA^2$ outcomes otherwise.

\subsection{Preparations}

Before proceeding to the proof, we make some preparations.
From Theorem~1 of Ref.~\cite{Chi-Dar-Sch-2007}, any $\hA \in \POVMA$ can be expressed as
\begin{eqnarray}
 \hA(\omega) &=& \int_\Gamma \hE^{(\gamma)}(\omega) p(d\gamma), \label{eq:A_E}
\end{eqnarray}
where $\hE^{(\gamma)} \in \POVMA$ is a POVM with at most $\dA^2$ outcomes,
$\Gamma$ is the entire set of indices $\gamma$
such that $\hE^{(\gamma)}$ is a POVM with at most $\dA^2$ outcomes,
and $p$ is a probability measure, which satisfies $p(\Gamma) = 1$.
From Eqs.~\eqref{eq:PiA}, \eqref{eq:primal}, and \eqref{eq:A_E}, we have
\begin{eqnarray}
 f(\hA) &=& \sum_m \Tr \left[ \hc_m \int_\Omega
                        \left[ \int_\Gamma \hE^{(\gamma)}(\dw) p(d\gamma) \right] \otimes \hBw_m \right]
 \nonumber \\
 &=& \int_\Gamma f[\hE^{(\gamma)}] p(d\gamma). \label{eq:A_E_f}
\end{eqnarray}
Let us define $\Gamma^\c$ as
\begin{eqnarray}
 \Gamma^\c &\equiv& \{ \gamma \in \Gamma : \hE^{(\gamma)} \in \POVMA^\c \},
\end{eqnarray}
which is the entire set of indices $\gamma$ such that $\hE^{(\gamma)}$ is
a feasible solution to Problem~P.
Let $\hA$ be an optimal solution to Problem~P.

We show the following lemma:
\begin{lemma} \label{lemma:pGamma}
 If $p(\Gamma^\c) = 1$ holds, then there exists an optimal solution to Problem~P
 with at most $\dA^2$ outcomes.
\end{lemma}
\begin{proof}
 Let $\gamma^\opt$ be an index satisfying
 \begin{eqnarray}
  \gamma^\opt &\in& \argmax_{\gamma \in \Gamma^\c} f[\hE^{(\gamma)}].
 \end{eqnarray}
 From Eq.~\eqref{eq:A_E_f}, we have
 \begin{eqnarray}
  f^\opt &=& f(\hA) = \int_\Gamma f[\hE^{(\gamma)}] p(d\gamma) \nonumber \\
  &=& \int_{\Gamma^\c} f[\hE^{(\gamma)}] p(d\gamma)
   \le f[\hE^{(\gamma^\opt)}].
 \end{eqnarray}
 On the other hand, from $\gamma^\opt \in \Gamma^\c$ (i.e., $\hE^{(\gamma^\opt)} \in \POVMA^\c$),
 $f[\hE^{(\gamma^\opt)}] \le f^\opt$ must hold.
 Thus, $f[\hE^{(\gamma^\opt)}] = f^\opt$.
 Therefore, $\hE^{(\gamma^\opt)}$, which is a POVM with at most $\dA^2$ outcomes,
 is an optimal solution to Problem~P.
 \QED
\end{proof}

\subsection{Proof}

We first consider the case $J = 0$.
From $\Gamma^\c = \Gamma$, $p(\Gamma^\c) = 1$ holds.
Thus, from Lemma~\ref{lemma:pGamma}, there exists an optimal solution to Problem~P
with at most $\dA^2$ outcomes.
For the remainder of the proof, the case $J \ge 1$ is considered.

In the following, we will show that $f^\bullet(\beta)$ is a concave function.
It suffices to consider the range of $\beta$ such that $f^\bullet(\beta) > -\infty$.
Let $\POVMA^\bullet(\beta) \subseteq \POVMA$ be the feasible set of the optimization problem
obtained by replacing $b$ of Problem~P with $\beta$.
Now we consider distinct $\beta^{(1)}, \beta^{(2)} \in \Real^J$.
For each $k \in \{ 1, 2 \}$, there exists $\hA_k \in \POVMA^\bullet[\beta^{(k)}]$ satisfying
$f(\hA_k) = f^\bullet[\beta^{(k)}]$.
Since $t\hA_1 + (1-t)\hA_2 \in \POVMA^\bullet[t\beta^{(1)} + (1-t)\beta^{(2)}]$ with $0 \le t \le 1$ holds,
we obtain
\begin{eqnarray}
 \lefteqn{ f^\bullet[t \beta^{(1)} + (1-t) \beta^{(2)}] } \nonumber \\
 &\ge& f[t\hA_1 + (1-t)\hA_2] \nonumber \\
 &=& t f(\hA_1) + (1-t) f(\hA_2) \nonumber \\
 &=& t f^\bullet[\beta^{(1)}] + (1-t) f^\bullet[\beta^{(2)}].
\end{eqnarray}
Therefore, $f^\bullet(\beta)$ is concave.

Let us consider a linear function $f_{\rm L}(\beta)$ such that
\begin{eqnarray}
 f_{\rm L}(\beta) - f^\bullet(\beta) &\ge& f_{\rm L}(b) - f^\bullet(b) = 0.
\end{eqnarray}
Since $- f^\bullet(\beta)$ is convex and thus subdifferentiable at each point
\cite{Boy-2009},
there always exists such $f_{\rm L}(\beta)$.
Let
\begin{eqnarray}
 \mD &\equiv& \{ \beta \in \Real^J : f_{\rm L}(\beta) = f^\bullet(\beta) \}.
  \label{eq:D}
\end{eqnarray}
It follows that $\mD$ is a convex set including $b$.
Let $\mE_\mD$ be the entire set of extremal points of $\mD$.
Also, let $\mE$ be the entire set of $\beta' \in \Real^J$
such that $f^\bullet(\beta)$ is strictly concave at $\beta = \beta'$.
We can easily verify $\mE_\mD \subseteq \mE$.

First, we consider the case $b \in \mE_\mD$.
From $b \in \mE$, $f^\bullet(\beta)$ is strictly concave at $\beta = b$.
From Lemma~\ref{lemma:pGamma}, it suffices to show $p(\Gamma^\c) = 1$; 
assume by contradiction that $p(\Gamma^\c) < 1$.
Let, for each $j \in \mI_J$,
\begin{eqnarray}
 \hspace{-1em}
  \Gamma_j &\equiv& \{ \gamma \in \Gamma : \eta_k[\hE^{(\gamma)}] \le 0 ~(\forall k \in \mI_j),
  \eta_j[\hE^{(\gamma)}] > 0 \}. \nonumber \\
\end{eqnarray}
For simplicity, let $\Gamma_J \equiv \Gamma^\c$.
$\{ \Gamma_j \}_{j=0}^J$ are obviously disjoint sets satisfying $\bigcup_{j=0}^J \Gamma_j = \Gamma$.
Let $p_j \equiv p(\Gamma_j)$ and, for each $j \in \mI_{J+1}$,
\begin{eqnarray}
 \hA_j(\omega) &\equiv&
  \left\{
   \begin{array}{ll}
    \displaystyle   \int_{\Gamma_j} \hE^{(\gamma)}(\omega) \frac{p(d\gamma)}{p_j},
     & p_j > 0, \\
    0, & {\rm otherwise}; \\
   \end{array} \right. \label{eq:finite_Aj}
\end{eqnarray}
then, $\hA_j$ is in $\POVMA$ if $p_j > 0$ holds.
From Eqs.~\eqref{eq:A_E} and \eqref{eq:finite_Aj}, we have
\begin{eqnarray}
 \hA(\omega) &=& \sum_{j=0}^J \int_{\Gamma_j} \hE^{(\gamma)}(\omega) p(d\gamma) = \sum_{j=0}^J p_j \hA_j(\omega). \label{eq:A_Aj}
\end{eqnarray}
Thus, we obtain
\begin{eqnarray}
 f^\bullet(b) &=& f(\hA) = f\left( \sum_{j=0}^J p_j \hA_j \right) \nonumber \\
 &=& \sum_{j=0}^J p_j f(\hA_j) \le \sum_{j=0}^J p_j f^\bullet[\beta^{(j)}], \label{eq:finite_convex}
\end{eqnarray}
where $\beta^{(k)} \equiv \{ \eta_j(\hA_k) \}_{j=0}^{J-1}$.
The inequality follows from the fact that $\hA_j \in \POVMA^\bullet[\beta^{(j)}]$
(i.e., $f(\hA_j) \le f^\bullet[\beta^{(j)}]$) holds when $p_j > 0$.
On the other hand, it follows that $p_j < 1$ holds for any $j \in \mI_{J+1}$.
Indeed, $p_J = p(\Gamma^\c) < 1$ obviously holds.
Also, since $\eta_j[\hE^{(\gamma)}] > 0$ holds for any $j \in \mI_J$ and $\gamma \in \Gamma_j$,
if $p_j = 1$ holds for some $j \in \mI_J$, then, $\eta_j(\hA) = \eta_j(\hA_j) > 0$ holds from Eq.~\eqref{eq:finite_Aj},
which contradicts $\hA \in \POVMA^\c$.
Thus, there exist at least two distinct integers $k \in \mI_{J+1}$ satisfying $p_k > 0$.
This implies that, from Eq.~\eqref{eq:finite_convex},
$f^\bullet(\beta)$ is not strictly concave at $\beta = b$ (i.e., $b \not\in \mE$),
which contradicts $b \in \mE_\mD \subseteq \mE$.
Therefore, $p(\Gamma^\c) = 1$ must hold.
From Lemma~\ref{lemma:pGamma},
there exists an optimal solution to Problem~P with at most $\dA^2$ outcomes.

Next, we consider the case $b \not\in \mE_\mD$.
Since $\mD$ is convex, from the finite-dimensional version of
the Krein-Milman theorem \cite{Kre-Mil-1940},
$\mD$ is the convex hull of $\mE_\mD$.
Thus, from Carath\'{e}odory's theorem,
there exists a set of $J+1$ points $\{ b^{(j)} \}_{j=0}^J \subseteq \mE_\mD$
such that $b \in \mD$ lies in the convex hull of $\{ b^{(j)} \}$
(note that $b^{(j)}$ and $b^{(j')}$ $~(j \neq j')$ can be the same).
This indicates that there exists $\{ q_j \}_{j=0}^J \in \Realp^{J+1}$ with $\sum_{j=0}^J q_j = 1$
such that $b = \sum_{j=0}^J q_j b^{(j)}$.
From $b^{(j)} \in \mE_\mD$, similar to the above discussion,
it follows that, for each $j \in \mI_{J+1}$, there exists $\gamma_j \in \Gamma$ satisfying
$f[\hE^{(\gamma_j)}] = f^\bullet[b^{(j)}]$ and $\hE^{(\gamma_j)} \in \POVMA^\bullet[b^{(j)}]$.
Using such $\gamma_j$, let
\begin{eqnarray}
 \hA' &\equiv& \sum_{j=0}^J q_j \hE^{(\gamma_j)};
\end{eqnarray}
then, we have
\begin{eqnarray}
 f(\hA') &=& f \left[ \sum_{j=0}^J q_j \hE^{(\gamma_j)} \right]
  = \sum_{j=0}^J q_j f [\hE^{(\gamma_j)}] \nonumber \\
 &=& \sum_{j=0}^J q_j f^\bullet [b^{(j)}]
  = f^\bullet(b) = f^\opt.
\end{eqnarray}
The fourth equality follows from the fact that,
from $b, b^{(j)} \in \mD$ and Eq.~\eqref{eq:D},
$f^\bullet(b) = f_{\rm L}(b)$ and $f^\bullet[b^{(j)}] = f_{\rm L}[b^{(j)}]$ hold,
and that $f_{\rm L}(\beta)$ is linear.
Also, from $\hE^{(\gamma_j)} \in \POVMA^\bullet[b^{(j)}]$,
$\hA' \in \POVMA^\bullet(b) = \POVMA^\c$ holds.
Thus, $\hA'$, which is a POVM with at most $(J+1)\dA^2$ outcomes,
is an optimal solution to Problem~P.
\QED

\section{Proof of Theorem~\ref{thm:sym}} \label{append:sym}

Using Eq.~\eqref{eq:group_action_S}, we can easily verify that,
the following equations hold for any $g \in \mG$, $c \in \Real$, and $\hQ, \hR \in \mS$:
\begin{eqnarray}
 g \c (\hQ + \hR) &=& g \c \hQ + g \c \hR, \nonumber \\
 g \c (\hQ\hR) &=& (g \c \hQ)(g \c \hR), \nonumber \\
 g \c (c \hQ) &=& c (g \c \hQ), \label{eq:g_c} \nonumber \\
 g \c \hat{1} &=& \hat{1}, \label{eq:g_1} \nonumber \\
 \Tr(g \c \hQ) &=& \Tr~\hQ, \nonumber \\
 g \c \hQ &\ge& 0, ~~ \forall \hQ \ge 0, \nonumber \\
 g \c \TrB~\hQ &=& \TrB(g \c \hQ). \label{eq:sym_form1}
\end{eqnarray}
The similar equations (except for the last one) for $\mSA$ and $\mSB$ instead of $\mS$ also hold.
Also, from Eqs.~\eqref{eq:UVW} and \eqref{eq:group_action_SA},
we have that for any $\hQ^{(\A)} \in \mSA$ and $\hQ^{(\B)} \in \mSB$,
\begin{eqnarray}
 g \c [\hQ^{(\A)} \otimes \hQ^{(\B)}]
  &=& [g \c \hQ^{(\A)}] \otimes [g \c \hQ^{(\B)}]. \label{eq:sym_form2}
\end{eqnarray}
In what follows, we will often make use of these equations without mentioning it.

Let, for any $g \in \mG$ and $\hPhi \in \POVMA^\c$,
\begin{eqnarray}
 \hPhi^\g(\omega) &\equiv& \inv{g} \c \hPhi(g \c \omega). \label{eq:sym_Phi}
\end{eqnarray}
From Eq.~\eqref{eq:sym_B}, we obtain
\begin{eqnarray}
 g \c \hPi^{(\hPhi^\g)}_m &=& \int_\Omega [g\c \hPhi^\g(\dw)] \otimes \left[ g \c \hBw_m \right] \nonumber \\
 &=& \int_\Omega \hPhi[d(g \c \omega)] \otimes \hB^{(g \c \omega)}_{g \c m} \nonumber \\
 &=& \hPi^{(\hPhi)}_{g \c m}. \label{eq:sym_g_Pi}
\end{eqnarray}

We first show that the mapping $\kappa_g: \hPhi \mapsto \hPhi^\g$ is
bijective on $\POVMA^\c$.
We can easily verify that $\hPhi^\g$ is a POVM.
We have that for any $j \in \mI_J$,
\begin{eqnarray}
 \summ \Tr \left[ \ha_{j,m} \hPi^{(\hPhi^\g)}_m \right]
 &=& \summ \Tr \left[ (g \c \ha_{j,m}) \left[ g \c \hPi^{(\hPhi^\g)}_m \right] \right] \nonumber \\
 &=& \summ \Tr \left[ \ha_{g \c j,g \c m} \hPi^{(\hPhi)}_{g \c m} \right] \nonumber \\
 &=& \sum_{m'=0}^{M-1} \Tr \left[ \ha_{g \c j,m'} \hPi^{(\hPhi)}_{m'} \right] \nonumber \\
 &\le& b_{g \c j} = b_j, \label{eq:sym_A_Phig}
\end{eqnarray}
where $m' = g \c m$.
The second and fourth lines follow from Eq.~\eqref{eq:sym_g_Pi} and $\hPhi \in \POVMA^\c$, respectively.
Thus, $\hPhi^\g$ is in $\POVMA^\c$.
Also, $\kappa_{\inv{g}}$ is the inverse mapping of $\kappa_g$.
Therefore, $\kappa_g$ is bijective on $\POVMA^\c$.

We next define $\hA \in \POVMA^\c$ as
\begin{eqnarray}
 \hAw &\equiv& \frac{1}{|\mG|} \sumg \hPhi^\g(\omega), \label{eq:sym_A_Phi}
\end{eqnarray}
and show Eq.~\eqref{eq:sym_A}, $\hA \in \POVMA^\c$, and $f(\hA) = f(\hPhi)$.
We have that for any $g \in \mG$ and $m \in \mI_M$,
\begin{eqnarray}
 g \c \hAw &=& \frac{1}{|\mG|} \sum_{h \in \mG} g \c \hPhi^{(h)}(\omega) \nonumber \\
 &=& \frac{1}{|\mG|} \sum_{h' \in \mG} \inv{h'} \c \hPhi(h' \c g \c \omega) \nonumber \\
 &=& \frac{1}{|\mG|} \sum_{h' \in \mG} \hPhi^{(h')}(g \c \omega)
 = \hA(g \c \omega), \label{eq:sym_gA}
\end{eqnarray}
where $h' = h \c \inv{g}$.
This gives Eq.~\eqref{eq:sym_A}.
From Eq.~\eqref{eq:sym_A_Phig}, we have that for any $j \in \mI_J$,
\begin{eqnarray} \hspace{-1em}
 \summ \Tr \left[ \ha_{j,m} \hPiA_m \right]
  &=& \frac{1}{|\mG|} \sumg \summ \Tr \left[ \ha_{j,m} \hPi^{(\hPhi^\g)}_m \right] \le b_j.
  \nonumber \\
\end{eqnarray}
Thus, $\hA \in \POVMA^\c$ holds.
Moreover, we obtain
\begin{eqnarray}
 f(\hA) &=& \summ \Tr \left[ \hc_m \hPiA_m \right] \nonumber \\
 &=& \frac{1}{|\mG|} \sumg \summ \Tr \left[ \hc_m \hPi^{(\hPhi^\g)}_m \right] \nonumber \\
 &=& \frac{1}{|\mG|} \sumg \summ \Tr
  \left[ (g \c \hc_m) \left[ g \c \hPi^{(\hPhi^\g)}_m \right] \right] \nonumber \\
 &=& \frac{1}{|\mG|} \sumg \summ \Tr
  \left[ \hc_{g \c m} \hPi^{(\hPhi)}_{g \c m} \right] \nonumber \\
 &=& \frac{1}{|\mG|} \sumg f(\hPhi) \nonumber = f(\hPhi),
\end{eqnarray}
where the fourth line follows from Eq.~\eqref{eq:sym_g_Pi}.
In particular, if $\hPhi$ is an optimal solution to Problem~P,
then so is $\hA$.

We finally show that there exists $(\hX, \lambda) \in \mX^\c$ satisfying Eq.~\eqref{eq:sym_Xlambda}.
Let
\begin{eqnarray}
 \hY^\g &\equiv& g \c \hY, \nonumber \\
 \nu^\g &\equiv& \{ \nu_j^\g \equiv \nu_{\inv{g} \c j} \}_{j=0}^{J-1};
\end{eqnarray}
then, we have that for any $g \in \mG$ and $m \in \mI_M$,
\begin{eqnarray}
 \hspace{-1em}
 g \c \hz_m(\nu) &=& \hc_{g \c m} - \sumj \nu_j \ha_{g \c j,g \c m} \nonumber \\
 &=& \hc_{g \c m} - \sumj \nu_{g \c j}^\g \ha_{g \c j,g \c m} \nonumber \\
 &=& \hz_{g \c m}[\nu^\g].
\end{eqnarray}
Thus, we have that for any $\omega \in \Omega$,
\begin{eqnarray}
 \hY^\g &\ge& g \c \hsw(\nu) \nonumber \\
 &=& \TrB \summ \left[ g \c \hz_m(\nu) \right] \left[ g \c \hBw_m \right] \nonumber \\
 &=& \TrB \summ \hz_{g \c m}[\nu^\g] \hB^{(g \c \omega)}_{g \c m} \nonumber \\
 &=& \hsigma_{g \c \omega}[\nu^\g], \label{eq:sym_Xl1}
\end{eqnarray}
i.e., $[\hY^\g, \nu^\g] \in \mX^\c$.
Also, we obtain
\begin{eqnarray}
 s[\hY^\g, \nu^\g] &=& \Tr~\hY^\g + \sumj \nu^\g_j b_j \nonumber \\
 &=& \Tr~\hY + \sumj \nu_{\inv{g} \c j} b_{\inv{g} \c j} \nonumber \\
 &=& s(\hY, \nu). \label{eq:sym_Xl2}
\end{eqnarray}
Let us define $(\hX, \lambda)$ as
\begin{eqnarray}
 \hX &\equiv& \frac{1}{|\mG|} \sumg \hY^\g, ~~~
  \lambda_j \equiv \frac{1}{|\mG|} \sumg \nu_j^\g. \label{eq:sym_Xlambda2}
\end{eqnarray}
We can easily verify that Eq.~\eqref{eq:sym_Xlambda} holds.
We have that for any $\omega \in \Omega$,
\begin{eqnarray}
 \hsw(\lambda) &=& \frac{1}{|\mG|} \TrB \sumg \summ \hz_m[\nu^\g] \hBw_m \nonumber \\
 &=& \frac{1}{|\mG|} \sumg \hsw[\nu^\g],
\end{eqnarray}
which gives
\begin{eqnarray}
 \hX - \hsw(\lambda) &=& \frac{1}{|\mG|} \sumg [\hY^\g - \hsw[\nu^\g]] \ge 0,
\end{eqnarray}
i.e., $(\hX, \lambda) \in \mX^\c$.
Moreover, from Eqs.~\eqref{eq:sym_Xl2} and \eqref{eq:sym_Xlambda2},
we obtain
\begin{eqnarray}
 \hspace{-1em}
 s(\hX, \lambda) &=& \Tr~\hX + \sumj \lambda_j b_j \nonumber \\
 &=& \frac{1}{|\mG|} \sumg \left[ \Tr~\hY^{(b)} + \sumj \nu_j^\g b_j \right] \nonumber \\
 &=& \frac{1}{|\mG|} \sumg s[\hY^\g, \nu^\g] = s(\hY, \nu).
\end{eqnarray}
In particular, if $(\hY, \nu)$ is an optimal solution to Problem~DP,
then so is $(\hX, \lambda)$.
\QED

\section{Proof of Theorem~\ref{thm:sym_minimax}} \label{append:sym_minimax}

Let $(\zeta, \hPhi)$ be a minimax solution to Problem~$\Pm$.
Also, let $\mu^\opt \equiv \{ \mu^\opt_k \}_{k=0}^{K-1}$ with
\begin{eqnarray}
 \mu^\opt_k &\equiv& \frac{1}{|\mG|} \sum_{g \in \mG} \zeta_{g \c k}.
\end{eqnarray}
We can see that $\mu^\opt \in \Prob$ and the first line of Eq.~\eqref{eq:sym_minimax} hold.
Moreover, similar to Eq.~\eqref{eq:sym_A_Phi},
let $\hA^\opt(\omega) \equiv |\mG|^{-1} \sum_{g \in \mG} \hPhi^\g(\omega)$,
where $\hPhi^\g$ is defined by Eq.~\eqref{eq:sym_Phi};
then, from Eq.~\eqref{eq:sym_gA}, the second line of Eq.~\eqref{eq:sym_minimax} holds.
The only thing we have to show now is that $(\mu^\opt, \hA^\opt)$ is also a minimax solution.
From Theorem~\ref{thm:minimax},
it suffices to show $f_k(\hA^\opt) \ge F^\opt(\mu^\opt)$ for any $k \in \mI_K$.
In what follows, we will show $f_k(\hA^\opt) \ge F^\opt(\zeta)$ and $F^\opt(\zeta) \ge F^\opt(\mu^\opt)$.

First, we show $f_k(\hA^\opt) \ge F^\opt(\zeta)$ for any $k \in \mI_K$.
We have that for any $k \in \mI_K$,
\begin{eqnarray}
 f_k(\hA^\opt) &=& \frac{1}{|\mG|} \summ \sum_{g \in \mG}
  \Tr \left[ \hc_{k,m} \hPi^{(\hPhi^\g)}_m \right] + d_k \nonumber \\
 &=& \frac{1}{|\mG|} \summ \sum_{g \in \mG}
  \Tr \left[ (g \c \hc_{k,m}) \hPi^{(\hPhi)}_{g \c m} \right] + d_k \nonumber \\
 &=& \frac{1}{|\mG|} \sum_{m'=0}^{M-1} \sum_{g \in \mG}
  \Tr \left[ \hc_{g \c k,m'} \hPi^{(\hPhi)}_{m'} \right] + d_k \nonumber \\
 &=& \frac{1}{|\mG|} \sum_{g \in \mG} \left[ \sum_{m'=0}^{M-1}
  \Tr \left[ \hc_{g \c k,m'} \hPi^{(\hPhi)}_{m'} \right] + d_{g \c k} \right] \nonumber \\
 &=& \frac{1}{|\mG|} \sum_{g \in \mG} f_{g \c k}(\hPhi) \ge F^\opt(\zeta),
\end{eqnarray}
where $m' = g \c m$.
The second line follows from Eq.~\eqref{eq:sym_g_Pi}.
The inequality follows from the fact that, from Theorem~\ref{thm:minimax},
$f_k(\hPhi) \ge F^\opt(\zeta)$ holds for any $k \in \mI_K$.

Next, we show $F^\opt(\zeta) \ge F^\opt(\mu^\opt)$.
Let $\zeta^\g \equiv \{ \zeta^\g_k \equiv \zeta_{g \c k} \}_{k=0}^{K-1}$;
then, we have that for any $g \in \mG$,
\begin{eqnarray}
 F^\opt[\zeta^\g] &=& \max_{\Phi \in \POVMA^\c} \sum_{k=0}^{K-1} \zeta_{g \c k}
  \left[ \summ \Tr \left[ \hc_{k,m} \hPi^{(\hPhi)}_m \right] + d_k \right] \nonumber \\
 &=& \max_{\Phi \in \POVMA^\c} \sum_{k'=0}^{K-1} \zeta_{k'}
  \left[ \sum_{m'=0}^{M-1} \Tr \left[ \hc_{k',m'} \hPi^{(\hPhi^{(\inv{g})})}_{m'} \right] + d_{k'} \right] \nonumber \\
 &=& F^\opt(\zeta), \label{eq:sym_minimax_Fstar_eta}
\end{eqnarray}
where $k' = g \c k$ and $m' = g \c m$.
From Eq.~\eqref{eq:sym_minimax_Fstar_eta}, we obtain
\begin{eqnarray}
 F^\opt(\mu^\opt) &=& \max_{\Phi \in \POVMA^\c} \frac{1}{|\mG|} \sum_{g \in \mG} \sum_{k=0}^{K-1} \zeta^\g_k f_k(\Phi)
  \nonumber \\
 &\le& \frac{1}{|\mG|} \sum_{g \in \mG} \max_{\Phi \in \POVMA^\c} \sum_{k=0}^{K-1} \zeta^\g_k f_k(\Phi)
  \nonumber \\
 &=& \frac{1}{|\mG|} \sum_{g \in \mG} F^\opt[\zeta^\g] = F^\opt(\zeta);
\end{eqnarray}
thus, $(\mu^\opt, \hA^\opt)$ is a minimax solution.
\QED

\section{Derivation of $(\hX^\opt, \lambda^\opt)$} \label{append:dtrine}

We will obtain an optimal solution $(\hX^\opt, \lambda^\opt)$ to the problem
of Eq.~\eqref{eq:inc_dual}.
This can be derived by extending methods described in Refs.~\cite{Chi-Hsi-2013,Cro-Bar-Wei-2017},
in which a minimum-error sequential measurement for double trine states is obtained.

Now, we consider the problem of Eq.~\eqref{eq:inc_dual} in which $\lambda$ is fixed.
An optimal solution, denoted as $\hX^\opt(\lambda)$, to this problem can be expressed by
$\hX^\opt(\lambda) = \upsilon(\lambda) \identA$,
where $\upsilon(\lambda)$ is a real-valued function of $\lambda$.
It follows that $\upsilon(\lambda)$ is the minimum value satisfying
$\upsilon(\lambda) \identA \ge \hsw(\lambda)$ for any $\omega \in \Omega$,
which means that $\upsilon(\lambda)$ is the maximum value of
the largest eigenvalues of $\hsw(\lambda)$ over all $\omega \in \Omega$.

Substituting Eq.~\eqref{eq:dtrine_phi} into Eq.~\eqref{eq:inc_hsw} gives
\begin{eqnarray}
 \hsw(\lambda) &=& \sum_{m=0}^2 l^\w_m \ket{\phi_m} \bra{\phi_m}, \label{eq:dtrine_hsw}
\end{eqnarray}
where
\begin{eqnarray}
 l^\w_m &\equiv& \frac{1}{3} \braket{\phi_m | \left[ \hBw_m + \lambda \hBw_3 \right] | \phi_m}.
  \label{eq:dtrine_l}
\end{eqnarray}
Let $\upsilon_\omega^+(\lambda)$ and $\upsilon_\omega^-(\lambda)$ be
the eigenvalues of $\hsw(\lambda)$ with $\upsilon_\omega^+(\lambda) \ge \upsilon_\omega^-(\lambda)$.
$\hU_\theta$ is defined as
\begin{eqnarray}
 \hU_\theta &\equiv& (\cos \theta) \ident + \sin \theta (\ket{1}\bra{0} - \ket{0}\bra{1}),
  \label{eq:dtrine_u}
\end{eqnarray}
which is a unitary operator corresponding to a rotation of $\theta$.
There exists $\theta$ such that
\begin{eqnarray}
 \hsw(\lambda) &=&
  \hU_{\frac{\theta}{2}} [ \upsilon_\omega^-(\lambda) \ket{0}\bra{0} +
  \upsilon_\omega^+(\lambda) \ket{1}\bra{1} ] \hU_{\frac{\theta}{2}}^\dagger.
  \label{eq:dtrine_schatten}
\end{eqnarray}
Using Eqs.~\eqref{eq:dtrine_hsw}, \eqref{eq:dtrine_u}, and \eqref{eq:dtrine_schatten}
and doing some algebra gives
\begin{eqnarray}
 \upsilon_\omega^+(\lambda) &=&
  \sum_{m=0}^2 \frac{1}{2} \left[ 1 - \cos \left( \theta - \frac{2\pi m}{3} \right) \right] l^\w_m.
  \label{eq:dtrine_up}
\end{eqnarray}
Substituting Eq.~\eqref{eq:dtrine_l} into Eq.~\eqref{eq:dtrine_up} yields
\begin{eqnarray}
 \upsilon_\omega^+(\lambda) &=& \frac{\lambda + 1}{2} \sum_{m=0}^3 \Tr \left[ \hrho^{(\theta)}_m \hBw_m \right],
  \label{eq:dtrine_up_tr}
\end{eqnarray}
where
\begin{eqnarray}
 \hrho^{(\theta)}_m &\equiv&
  \left\{
   \begin{array}{ll}
    \displaystyle \frac{1 - \cos \left( \theta - \frac{2\pi m}{3} \right)}{3(\lambda + 1)}
     \ket{\phi_m} \bra{\phi_m}, & m < 3, \\
    \displaystyle \lambda \sum_{r=0}^2 \hrho^{(\theta)}_r, & m = 3. \\
   \end{array} \right.
  \label{eq:dtrine_rhotheta}
\end{eqnarray}
We can easily see $\sum_{m=0}^3 \Tr~\hrho^{(\theta)}_m = 1$.

$\Tr[\hrho^{(\theta)}_m \hBw_m]$ in Eq.~\eqref{eq:dtrine_up_tr}
equals the average success probability of the POVM $\{ \hBw_m \}_{m=0}^3$
for the quaternary states $\{ \hrho^{(\theta)}_m \}_{m=0}^3$.
Let $P^\opt_\theta$ be the average success probability of
a minimum-error measurement for $\{ \hrho^{(\theta)}_m \}$;
then, from Eq.~\eqref{eq:dtrine_up_tr}, we have
\begin{eqnarray}
 \upsilon_\omega^+(\lambda) &\le& \frac{\lambda + 1}{2} P^\opt_\theta.
\end{eqnarray}
This gives
\begin{eqnarray}
 \upsilon(\lambda) &=& \max_\omega \upsilon_\omega^+(\lambda)
  \le \frac{\lambda + 1}{2} \max_\theta P^\opt_\theta.
  \label{eq:dtrine_u_P}
\end{eqnarray}
By the symmetry of the problem, we may, without loss of generality,
consider only the case $0 \le \theta \le \pi/3$
(i.e., $\Tr~\hrho^{(\theta)}_0 \le \Tr~\hrho^{(\theta)}_2 \le \Tr~\hrho^{(\theta)}_1$).
Using the method described in Ref.~\cite{Dec-Ter-2010}
(the method of Ref.~\cite{Ros-Mar-Gio-2017-binary} can also be used),
we can obtain an analytical expression of $P^\opt_\theta$ for each $\theta$.
To avoid cumbersome details, we do not give an analysical expression of $P^\opt_\theta$,
but note that $P^\opt_\theta$ achieves its maximum value if and only if $\theta = 0$ holds
and satisfies
\begin{eqnarray}
 P^\opt_\theta &\le& P^\opt_0 = 
  \left\{
   \begin{array}{ll}
    \displaystyle \frac{2 + \sqrt{3}}{4(\lambda+1)}, & \displaystyle \lambda \le \frac{1}{2} + \frac{1}{2\sqrt{3}}, \\
    \displaystyle \frac{\lambda(3\lambda-1)}{2(\lambda+1)(2\lambda-1)}, & {\rm otherwise}, \\
   \end{array} \right. \label{eq:dtrine_Popt}
\end{eqnarray}
where we assume $\lambda \le 1$ to simplify the discussion
(it is sufficient to consider only this case, as will be described in the main text).
From Eqs.~\eqref{eq:dtrine_u_P} and \eqref{eq:dtrine_Popt}, we have
\begin{eqnarray}
 \hspace{-1em}
 \upsilon(\lambda) &\le& \frac{\lambda + 1}{2} P^\opt_0 =
  \left\{
   \begin{array}{ll}
    \displaystyle \frac{2 + \sqrt{3}}{8}, & \displaystyle \lambda \le \frac{1}{2} + \frac{1}{2\sqrt{3}}, \\
    \displaystyle \frac{\lambda(3\lambda-1)}{4(2\lambda-1)}, & {\rm otherwise}. \\
   \end{array} \right. \label{eq:dtrine_upsilon}
\end{eqnarray}
A minimum-error measurement, denoted as $\{ \hB^\bullet_m \}_{m=0}^3$,
for the states $\{ \hrho^{(0)}_m \}_{m=0}^3$ (i.e., in the case of $\theta = 0$)
is given by
\begin{eqnarray}
 \hB^\bullet_m &=& \ket{B_m}\bra{B_m}, \nonumber \\
 \ket{B_0} &=& 0, \nonumber \\
 \ket{B_1} &=& \sqrt{\frac{1}{2}} \ket{0} - \sqrt{\frac{1-\alpha}{2}} \ket{1}, \nonumber \\
 \ket{B_2} &=& \sqrt{\frac{1}{2}} \ket{0} + \sqrt{\frac{1-\alpha}{2}} \ket{1}, \nonumber \\
 \ket{B_3} &=& \sqrt{\alpha} \ket{1}, \label{eq:dtrine_B}
\end{eqnarray}
where
\begin{eqnarray}
 \alpha &=& \frac{2(6\lambda^2 - 6\lambda + 1)}{3(2\lambda - 1)^2}.
\end{eqnarray} 

Let $\omega_0$ be in $\Omega$ such that $\{ \hB^{(\omega_0)}_m \} = \{ \hB^\bullet_m \}$.
In the case of $\omega = \omega_0$,
from Eqs.~\eqref{eq:dtrine_hsw} and \eqref{eq:dtrine_l},
Eq.~\eqref{eq:dtrine_schatten} with $\theta = 0$ holds.
Since, in this case, $\upsilon_{\omega_0}^+(\lambda) = \frac{\lambda + 1}{2} P^\opt_0$ holds,
the equality in Eq.~\eqref{eq:dtrine_upsilon} holds.
By substituting this into Eq.~\eqref{eq:inc_dual} and optimizing $\lambda$,
we obtain Eq.~\eqref{eq:dtrine_lambda_opt}.
From Eq.~\eqref{eq:dtrine_lambda_opt}, in the case of $\lambda = \lambda^\opt$, we have
\begin{eqnarray}
 \alpha &=& \frac{4\pinc}{3}.
\end{eqnarray}

\input{report-en.bbl}

\end{document}

%% file: settings_quant.tex
\newtheorem{thm}{Theorem}

\newtheorem{remark}[thm]{Remark}

\newtheorem{lemma}[thm]{Lemma}

\newtheorem{proof}{Proof}

\newcommand{\hA}{\hat{A}}
\newcommand{\hB}{\hat{B}}

\newcommand{\hE}{\hat{E}}

\newcommand{\hQ}{\hat{Q}}
\newcommand{\hR}{\hat{R}}

\newcommand{\hU}{\hat{U}}
\newcommand{\hV}{\hat{V}}
\newcommand{\hW}{\hat{W}}
\newcommand{\hX}{\hat{X}}
\newcommand{\hY}{\hat{Y}}

\newcommand{\ha}{\hat{a}}

\newcommand{\hc}{\hat{c}}

\newcommand{\hx}{\hat{x}}
\newcommand{\hy}{\hat{y}}
\newcommand{\hz}{\hat{z}}

\newcommand{\hPi}{\hat{\Pi}}
\newcommand{\hPhi}{\hat{\Phi}}

\newcommand{\hrho}{\hat{\rho}}

\newcommand{\hsigma}{\hat{\sigma}}

\newcommand{\mD}{\mathcal{D}}
\newcommand{\mE}{\mathcal{E}}

\newcommand{\mG}{\mathcal{G}}
\newcommand{\mH}{\mathcal{H}}
\newcommand{\mI}{\mathcal{I}}

\newcommand{\mM}{\mathcal{M}}

\newcommand{\mP}{\mathcal{P}}

\newcommand{\mS}{\mathcal{S}}
\newcommand{\mT}{\mathcal{T}}

\newcommand{\mW}{\mathcal{W}}
\newcommand{\mX}{\mathcal{X}}

\newcommand{\mZ}{\mathcal{Z}}

\newcommand{\trho}{\tilde{\rho}}

\newcommand{\PS}{P_{\rm S}}

\newcommand{\ident}{\hat{1}}

\newcommand{\Real}{\mathbf{R}}

\newcommand{\POVM}{\mM}

\newcommand{\QED}{\hspace*{0pt}\hfill $\blacksquare$}

\newcommand{\argmax}{\mathop{\rm argmax}}

\newcommand{\Tr}{{\rm Tr}}

\def\gauss_sym#1{{\lfloor #1 \rfloor}}

%% file: report-en.bbl
%